\documentclass[final,11pt]{article}

\usepackage{amsmath,amssymb,amsfonts,amsthm} 
\usepackage{epsfig} \usepackage{latexsym,nicefrac,bbm}
\usepackage{xspace}
\usepackage{color,fancybox,graphicx,url,subfigure,fullpage}
\usepackage[top=1in, bottom=1in, left=1in, right=1in]{geometry}
\usepackage{tabularx} \usepackage{hyperref} 
\usepackage{enumitem}
\usepackage{comment}
\usepackage{url}
\usepackage{enumitem}
\usepackage{booktabs,algorithmic}
\usepackage{cool}
\usepackage{bbm}
\usepackage{commath}

\Style{DSymb={\mathrm d},DShorten=true,IntegrateDifferentialDSymb=\mathrm{d}}

\usepackage{mathptmx}
\DeclareMathAlphabet{\mathcal}{OMS}{cmsy}{m}{n}

\usepackage[section]{algorithm}

\newtheorem{theorem}{Theorem}[section]
\newtheorem{lemma}[theorem]{Lemma}
\newtheorem{definition}[theorem]{Definition}

\newtheorem{proposition}[theorem]{Proposition}

\newcommand{\defeq}{\stackrel{\textup{def}}{=}}

\newcommand\pr{\mathop{\mathbb{P}}}
\newcommand\av{\mathop{\mathbb{E}}}

\newcommand{\nfrac}[2]{\nicefrac{#1}{#2}}
\def\abs#1{\left| #1 \right|}
\renewcommand{\norm}[1]{\ensuremath{\left\lVert #1 \right\rVert}}

\newcommand{\floor}[1]{\left\lfloor\, {#1}\,\right\rfloor}
\newcommand{\ceil}[1]{\left\lceil\, {#1}\,\right\rceil}

\newcommand\bz{\mathbb Z}
\newcommand\nat{\mathbb N}
\newcommand\rea{\mathbb R}

\newcommand\poly{\mathrm{poly}}

\newcommand\calI{\mathcal{I}}

\newcommand\calK{\mathcal{K}}

\renewcommand{\epsilon}{\varepsilon}
\newcommand{\eps}{\varepsilon}
\renewcommand{\tilde}{\widetilde}
\newcommand{\Span}{{\ensuremath{\mathsf{Span}}\xspace}}

\renewcommand{\Pr}{\pr}

\renewcommand{\hat}{\widehat}

\renewcommand{\vec}[1]{{\bf {#1}}}

\title{\bf Approximation Theory and  the Design of 
   Fast Algorithms}

\date{}

\author{Sushant Sachdeva\thanks{Research Fellow, Simons Institute for
    the Theory of Computing. UC Berkeley, USA. Part of this work was
    done when this author was a graduate student at the Department of
    Computer Science, Princeton University. Email: {\tt
      sachdeva@eecs.berkeley.edu}} \and Nisheeth
  K. Vishnoi\thanks{Microsoft Research, Bangalore, India. Email: {\tt
      nisheeth.vishnoi@gmail.com}}}
\begin{document}

\maketitle

\vspace{-5mm}

\begin{abstract}
We survey  key techniques and results from approximation theory in the context of uniform approximations to real functions such as $e^{-x},  \nfrac{1}{x}$,
and $x^k.$ We then present a selection of results demonstrating how  such approximations can be used to speed up primitives crucial for the design of fast algorithms for problems such as simulating random walks, graph partitioning, solving linear system of equations, computing eigenvalues and combinatorial approaches to solve semi-definite programs. 

\end{abstract}

\vspace{-5mm}
\tableofcontents 

\newpage

\section{Introduction}

\paragraph{A brief history of approximation theory.}
The area of approximation theory is concerned with the study of how  well functions can be approximated by simpler ones. While there are several notions  of {\em well} and {\em simpler}, arguably, the most natural notion is that of uniform approximations by polynomials: given a function $f:\mathbb{R} \mapsto \mathbb{R},$ and an interval  $\mathcal{I},$ what is the closest a degree $d$ polynomial can remain to  $f(x)$ in the entire interval? Formally, if $\Sigma_d$ is the class of all univariate real polynomials of degree at most $d,$ the goal is to understand
$$\eps_{f,\mathcal{I}}(d) \defeq \inf_{p \in \Sigma_d} \sup_{x \in \mathcal{I}} | f(x)-p(x) |.$$
This notion of approximation, sometimes called Chebyshev
approximation, is attributed to Pafnuty Chebyshev, who essentially
started this area in an attempt to improve upon the {\em parallel
  motion} invented by James Watt for his steam engine, see
\cite{Chebyshev1854}.  Chebyshev discovered the {\em alternation
  property} of the best approximating polynomial and computed the best
degree $d-1$ polynomial to approximate the monomial $x^{d},$ see
\cite{Chebyshev1859}.  The result equivalently showed that any degree
$d$ polynomial with leading coefficient $1$ cannot come more than
$\nfrac{1}{2^{d}}$ close to $0$ everywhere in the interval $[-1,1].$
Moreover, he showed that the degree $d$ polynomial that arises when
one writes $\cos (d \theta)$ as a polynomial in $\cos \theta$ achieves
this bound.  These polynomials are called Chebyshev polynomials which
find use in several different areas of science and mathematics and,
indeed, repeatedly make an appearance in this survey because of their
extremal properties.

\vspace{2mm}
Despite Chebyshev's seminal  results in approximation theory, including his work on best rational approximations, several foundational problems remained open. While it is obvious that $\eps_d(f,\mathcal{I})$  does not increase as we increase  $d,$ it was  Weierstrass~\cite{Weierstrass1885} who  later established that, for any continuous function $f$ and a bounded interval $\mathcal{I},$ $\eps_{f,\mathcal{I}}(d) \rightarrow 0$ as $d \rightarrow \infty.$ 
Further, it was Emile Borel~\cite{Borel1905} who proved that the best approximation is always {\em achieved} and is {\em unique}. Among other notable initial results in approximation theory, A. A. Markov~\cite{Markov1890}, motivated by a question in chemistry due to Mendeleev, proved that the absolute value of the derivative of a degree $d$ polynomial which is bounded by $1$ in the interval $[-1,1]$ cannot exceed $d^2.$ 
These, and several other results, not only solved important problems motivated by science and  engineering, but also significantly impacted theoretical areas such as mathematical analysis in the early 1900s. 

\vspace{2mm}
With computers coming into the foray around the mid 1900s, there was a fresh flurry of activity in the area of approximation theory. The primary goal was to come up with efficient ways to calculate mathematical functions arising in scientific computation and numerical analysis. For instance,  to evaluate $e^x$ for $x \in [-1,1],$ it is sufficient to store the coefficients of the best polynomial (or rational) approximation for it in this interval. 
For a fixed error, such approximations often provided a significantly more succinct representation of the function than the representation obtained by truncating the appropriate Taylor series. 

\vspace{2mm}
Among all this activity, an important development  happened in the 1960s  when  Donald Newman~\cite{Newman64} showed that the best degree-$d$ rational
approximation to the function $|x|$ on $[-1,1]$ achieves an
approximation error of $e^{-\Theta(\sqrt{d})},$ while  the  best degree-$d$ polynomial approximation only achieves an error of $\Theta(\nfrac{1}{d}).$
Though rational functions had also been considered by Chebyshev, it was Newman's result that revived the area of uniform approximation with rational functions and led to several results where the degree-error trade-off was exponentially better than that achievable by polynomial approximations.
Perhaps the  problem that received the most attention, due to its implications to numerical methods for solving systems of partial differential equations (see \cite{CodyMV69}), was to  understand the best rational approximation to $e^{-x}$ over the interval  $[0,\infty).$  
Note that $e^{-x}$ goes to $0$ as $x$ goes to infinity, while any polynomial must necessarily go to infinity.  
Rational functions of degree $d$ were shown to approximate $e^{-x}$ on $[0,\infty)$ up to an error of $c^d$ for $c<1.$ 
This line of research culminated in a landmark result by Gonchar and Rakhmanov~\cite{GoncharR83} who determined the optimal $c.$  
Despite remarkable progress in the theory of approximation by rational functions, unfortunately, there seems to be no clear understanding of why rational approximations are often significantly better than polynomials of the same degree,  and this area seems to be flooded with many surprising results often proven using clever tricks. Perhaps, this is what makes the area of rational approximations promising and worth understanding; it seems capable of magic.

\paragraph{Approximation theory in algorithms and complexity.}
Two of the first applications of approximation theory in algorithms\footnote{More precisely, in the area of numerical linear algebra, since algorithms was not yet established as a field.} were the Conjugate Gradient method \cite{HestenesS52} and the Lanczos method \cite{Lanczos52}, which are used to solve linear systems of equations $Ax=v$ where $A$ is an $n \times n$ real, symmetric and positive semi-definite (PSD) matrix. These results, which surfaced in the 1950s, resulted in what are called {\em Krylov subspace methods} and, can also be used to speed up eigenvalue and eigenvector (e.g., singular value/singular vector) calculations. These methods are iterative and reduce such computations to a small number of calculations of the form $Au$ for different vectors $u.$ Thus, they are particularly suited for sparse matrices that are too large to handled by Gaussian elimination-based methods; see the survey \cite{SaadSurvey} for a detailed discussion.

\vspace{2mm} Until recently, the main applications of approximation
theory in theoretical computer science have been in complexity theory:
one of the first being a result of Beigel {\em et al.}
\cite{BeigelRS95} who used Newman's result on rational approximations
to show that the complexity class PP is closed under intersections and
unions.\footnote{PP is the complexity class that contains sets which
  are accepted by a polynomial-time bounded probabilistic Turing
  machine that accepts with probability strictly more than
  $\nfrac{1}{2}.$} Another important result where approximation
theory, in particular Chebyshev polynomials, played a role is the {\em
  quadratic} speed-up for quantum search algorithms, initiated with a
work by Grover \cite{Grover96}. The fact that one cannot speed up
beyond Grover's result was showed by Beals {\em et al.} \cite{Beals01}
which, in turn, relied on the use of Markov's theorem as inspired by
Nisan and Szegedy's lower bound for the Boolean OR function
\cite{NisanS94}.  For more on applications of approximation theory to
complexity theory, communication complexity and computational learning
theory, we refer the reader to the bibliography by Aaronson
\cite{Aaronson08} and the thesis by Sherstov \cite{SherstovThesis}.

\vspace{2mm} In this survey, we present applications of approximation
theory to the design of fast algorithms. We show how to compute good
approximations to matrix-vector products such as $A^sv,$ $A^{-1}v$ and
$\exp({-A})v$ for any matrix $A$ and a vector $v.$\footnote{Recall
  that the matrix exponential is defined to be $\exp(A) \defeq \sum_{k
    \ge 0} \frac{A^k}{k!}$.}  Such primitives are useful for
performing several fundamental computations quickly, such as random
walk simulation, graph partitioning, solving linear system of
equations, and combinatorial approaches to solve semi-definite
programs.  The algorithms for computing these primitives end up
performing calculations of the form $Bu$ where $B$ is a matrix closely
related to $A$ (often $A$ itself) and $u$ is some vector. A key
feature of these algorithms is that if the matrix-vector product for
$A$ can be computed quickly, e.g., when $A$ is sparse, then $Bu$ can
also be computed in essentially the same time. This makes such
algorithms particularly relevant for handling the problem of {\em big
  data}. Such matrices capture either numerical data or large graphs,
and it is inconceivable to be able to compute much more than a few
matrix-vector product on matrices of this size.

\vspace{2mm}
As a simple but important application, we  show how to speed up the computation of  $A^sv$ where $A$ is a symmetric matrix with eigenvalues in $[-1,1],$ $v$ is a vector and $s$ is a large positive integer. 
The straightforward way to compute $A^sv$ takes time $O(ms)$ where $m$ is the number of non-zero entries in $A$, \emph{i.e.}, $A$'s \emph{sparsity}. We show how, appealing to a result from approximation theory, we can bring this running time down to essentially $O(m \sqrt{s}).$
We start with a result on polynomial approximation for $x^s$ over the
interval $[-1,1].$ Using some of the earliest results proved by
Chebyshev, it can be shown that there is a polynomial $p$ of degree $d
\approx \sqrt{s \log \nfrac{1}{\delta}}$ that $\delta$-approximates
$x^s$ over $[-1,1].$  
A straightforward diagonalization argument  then implies  that $\| A^s
v - \sum_{i=0}^d a_i A^i v\|_2 \leq \delta,$ where $p(x) = \sum_{i=0}
a_i x^i.$ 
More importantly, the time it takes to compute   $ \sum_{i=0}^d a_i A^i v$ is $O(md)=O(m \sqrt {s \log \nfrac{1}{\delta}}),$ which gives us a saving of about $\sqrt{s}.$ When $A$ is the random walk matrix of a graph and $v$ is an initial distribution over the vertices, the result above implies that we can speed up the computation of the distribution after $s$ steps by a quadratic factor. 
Note that  this application also motivates why uniform approximation
is the right notion for algorithmic applications, since all we know is
the interval in which eigenvalues of $A$ lie $v$ can be any vector
and, hence, we would like the approximating polynomial to be close
everywhere in that interval.

\vspace{2mm} While the computation of $\exp(-A)v$ is of fundamental
interest in several areas of mathematics, physics, and engineering,
our interest stems from its recent applications in algorithms and
optimization.  Roughly, these latter applications are manifestations
of the multiplicative weights method for designing fast algorithms,
and its extension to solving semi-definite programs via the framework
by Arora and Kale~\cite{AroraK07}.\footnote{See also
  \cite{JainJUW11,JainUW09,JainW09,Kthesis,OSVV,OrecchiaV11,
    OrecchiaSV11-arxiv, Vishnoi12, AHK12, Sherman09, Sherman13}.}  At the
heart of all algorithms based on the matrix multiplicative weights
update method is a procedure to quickly compute $e^{-A}v$ for a
symmetric, positive semi-definite matrix $A$ and a vector $v.$ Since
exact computation of the matrix exponential is expensive, we seek an
approximation. It suffices to approximate the function $e^{-x}$ on the
interval $[0,\|A\|].$ A simple approach is to truncate the Taylor
series expansion of $e^{-x}.$ It is easy to show that using roughly
$\|A\|+\log \nfrac{1}{\delta}$ terms in the expansion suffices to
obtain a $\delta$ approximation. However, we can use a polynomial
approximation result for $e^{-x}$ over the interval $[0,\|A\|]$ to
produce an algorithm that runs in time roughly $O(m \sqrt{
  \|A\|})$. In fact, when $A$ has more structure, we can go beyond the
square-root barrier.

\vspace{2mm} For fast graph algorithms, often the quantity of interest
is $e^{-L}v,$ where $L$ is the combinatorial Laplacian of a graph, and
$v$ is a vector. The vector $e^{-L}v$ can also be interpreted as the
resulting distribution of a certain continuous-time random walk on the
graph with starting distribution $v.$ Appealing to a rational
approximation to $e^{-x}$ with negative poles, the computation of
$e^{-L}v$ can be reduced to a small number of computations of the form
$L^{-1}u.$ Thus, using the near-linear-time Laplacian
solver\footnote{A Laplacian solver is an algorithm that
  (approximately) solves a given system of linear equations $Lx=b,$
  where $L$ is a graph Laplacian and $b\in {\rm Im}(L)$, \emph{i.e.},
  it (approximately) computes $L^{-1}b,$ see \cite{Vishnoi12}.} due to
Spielman and Teng~\cite{SpielmanT04}, this gives an
$\tilde{O}(m)$-time algorithm for approximating $e^{-L}v$ for graphs
with $m$ edges. In the language of random walks, continuous-time
random walks on an undirected graph can be simulated essentially
independent of time; such is the power of rational approximations.

\vspace{2mm}
A natural  question that arises from our last application is whether the Spielman-Teng result is necessary in order to compute $e^{-L}v$ in near-linear time. In our final application of approximation theory,  we answer this question in the
affirmative by presenting a reduction in the other direction; we prove that the inverse of a positive-definite matrix can be
  approximated by a weighted-sum of a small number of matrix
  exponentials.

\paragraph{Organization.}
The goal of this survey is to bring out how classical and modern
results from approximation theory play a crucial role in obtaining
results which are relevant today to the emerging theory of fast
algorithms.  The approach we have taken is to first present the ideas
and results from approximation theory that we think are central,
elegant, and have wider applicability in TCS. For the sake of clarity,
we have sometimes sacrificed tedious details. This means that we
admittedly either do not present complete proofs or theorems with
optimal parameters for a few important results. The survey is
organized as follows.

\vspace{2mm}
In Section \ref{sec:basics}, we present some essential notations and results from approximation theory, and introduce Chebyshev polynomials. We prove certain extremal properties of these polynomials which are used in this survey. In Section \ref{sec:poly-apx} we construct polynomial approximations to the monomial $x^s$ over the interval $[-1,1]$ and $e^{-x}$ over the interval $[0,b].$ Both results are  based on Chebyshev polynomials. In the same Section we prove a special case of Markov's theorem which is then used to show that these polynomial approximations are asymptotically optimal. 

\vspace{2mm}
In Section \ref{sec:rational-apx} we consider rational approximations for the function $e^{-x}$ over the interval $[0,\infty).$ We first show that degree $d$ rational functions can achieve $c^d$ error for some $0<c<1.$ Subsequently we prove that this result is optimal up to the choice of constant $c.$ 
In Section \ref{sec:exp-rational2} we present a proof of the
remarkable theorem that such geometrically decaying errors for the
$e^{-x}$ can be achieved by rational functions with an additional
restriction that all its poles be real and negative.  

\vspace{2mm}
In Section \ref{sec:exp-reduction} we show how $x^{-1}$ can be approximated by a sparse sum of the form $\sum_{i}w_i e^{-t_i x}$ over the interval $(0, 1].$ The proof relies on the Euler-Maclaurin formula and  certain bounds derived from the Riemann zeta function.

\vspace{2mm}
Section \ref{sec:applications} contains the presentation of
applications of the approximation theory results. In Section
\ref{sec:applications:sim-rw} we show how the results of Section
\ref{sec:poly-apx}  imply that we can quadratically speed up random
walks in graphs, and find sparse cuts faster. Here, we discuss the important issue of computing the coefficients of the polynomials in Section \ref{sec:poly-apx}. 
In Section \ref{sec:applications:cg} we present the famous Conjugate gradient method for solving symmetric PSD systems of equations $Ax=v$ iteratively where the number of iterations depend on the square-root of the condition number of $A.$ The square-root saving is shown to be exactly because of the scalar approximation result for $x^s$ from Section \ref{sec:poly-apx}.
In Section \ref{sec:lanczos} we present the  Lanczos method and show how it can be used to compute the largest eigenvalue of a symmetric matrix.  We show how  the  existence of a good approximation for $x^s,$ yet again, allows a quadratic speedup over the {\em power method}.

\vspace{2mm}
In Section  \ref{sec:applications:matrix-exp} we show how the polynomial and rational approximations to $e^{-x}$ developed in Section \ref{sec:rational-apx} imply the best known algorithms for computing $\exp(-A)v.$ If $A$ is a symmetric and diagonally dominant (SDD) matrix, then we show how to combine rational approximations to $e^{-x}$ with negative poles with the powerful SDD (Laplacian) solvers of Spielman-Teng to obtain near-linear time algorithms for computing  $\exp(-A)v.$ We also show how to bound and compute the coefficients involved in the rational approximation result in Section \ref{sec:exp-rational2}; this is crucial for the application.

\vspace{2mm}
Finally, in \ref{sec:applications:matrix-inverse-reduction}, using the result from Section \ref{sec:exp-reduction}, we show how to reduce computation of $A^{-1}v$ for a symmetric positive-definite (PD) matrix $A$ to the computation of a small number of  
computations of the form $\exp(-A)v.$ Apart from suggesting a new approach to solving a PD system, this result shows that computing $\exp(-A)v$  inherently requires the ability to solve a system of equations involving $A.$

\section{Basics}\label{sec:basics}
\subsection{Uniform Approximations}
\label{sec:poly:prelims}
Given an interval $\calI \subseteq \mathbb{R}$ and a function $f:\mathbb{R}\mapsto \mathbb{R},$ we are interested in 
\emph{approximations} for $f$ over $\calI.$ Of particular interest is the following notion of approximation.
\begin{definition}
A function $g$ is called a $\delta$-approximation to a
function $f$ over an interval $\calI$ if $\sup_{x \in \calI}
|f(x)-g(x)| \le \delta$. 
\end{definition}
\noindent
Both finite and infinite intervals
$\calI$ are considered. Such approximations are known as uniform approximations or Chebyshev approximations. 
Since then, a central topic of study in approximation theory has been to understand  how
well a function $f$ can be approximated using polynomials.
More precisely, the quantity of interest for a function $f$ is the best uniform error achievable over an interval $\mathcal{I}$ by a polynomial of degree $d,$ namely, $\eps_{f,\mathcal{I}}(d)$ as defined in the introduction. The first set of basic questions are 1) does $\lim_{d \rightarrow \infty} \eps_{f,\mathcal{I}}(d) =0$? and 2) does there always exist a degree-$d$ polynomial $p$ that achieves $\eps_{f,\mathcal{I}}(d)$? 
Interestingly, these questions were not addressed in Chebyshev's seminal work.
Later, 
Weierstrass~(see \cite{Rivlin69}) showed that, for a
continuous function $f$ on a bounded interval $[a,b],$ there exist
arbitrarily good polynomial approximations, \emph{i.e.}, for every
$\delta >0,$ there exists a polynomial $p$ that is a
$\delta$-approximation to $f$ on $[a,b].$ The existence and uniqueness of a a degree-$d$ polynomial that achieves the best approximation $\eps_{f,\mathcal{I}}(d)$  was proved by  Borel. 

The trade-off between the degree of the approximating polynomial and the approximation error has been
studied extensively, and is one of the main themes in this survey.

\vspace{2mm}

In an attempt to get a handle on best approximations, Chebyshev showed
that a degree-$d$ polynomial $p$ is a best approximation to $f$ over
an interval $[-1,1]$ if and only if the maximum error between $f$ and
$p$ is achieved exactly at $d+2$ points in $[-1,1]$ with alternating
signs, i.e, there are $-1 \leq x_0 < x_1 \cdots < x_{d+1} \leq 1$ such
that $f(x_i)-p(x_i)=(-1)^i \eps$ where $\eps = \sup_{x \in [-1,1]}
|f(x)-p(x)|.$ We prove the following theorem, attributed to de La
Vallee-Poussin, which implies the sufficient side of Chebyshev's
alternation theorem and is often good enough.
\begin{theorem}
\label{lem:poussin}
Suppose $f$ is a function over $[-1,1],$ $p$ is a degree-$d$
polynomial, and $\delta > 0$ is such that the \emph{error function}
$\eps \defeq f - p$ assumes alternately positive and negative signs at
$d+2$ increasing points $-1 \le x_0 < \cdots < x_{d+1} \le 1,$ and
satisfies $|\eps(x_i)| \ge \delta$ for all $i$. Then, for any degree-$d$ polynomial $q,$ we have $\sup_{x \in [-1,1]} |f(x) - q(x)| \ge
\delta.$
\end{theorem}
\begin{proof}
Suppose, to the contrary, that there exists a degree-$d$ polynomial $q$
such that $\sup_{x \in [-1,1]} |f(x) - q(x)| < \delta.$ This implies
that for all $i,$ we have $\eps(x_i) - \delta < q(x_i) - p(x_i) < \eps(x_i) +
\delta.$ Since $|\eps(x_i)| \ge \delta,$ the polynomial $q-p$ is
non-zero at each of the $x_i$s, and must have the same sign as
$\eps.$ Thus, $q-p$ assumes alternating signs at the $x_i$s, and
hence must have a zero between each pair of successive $x_i$s. This
implies that the non-zero degree-$d$ polynomial $q-p$ has at least
$d+1$ zeros, which is a contradiction.
\end{proof}

\noindent
The above theorem easily generalizes to any finite interval. In addition
to the conditions in the theorem, if we also have $\sup_{x \in [-1,1]}
|f(x) - p(x)| = \delta,$ then $p$ is the best degree-$d$
approximation. 
This theorem can be used to prove one of Chebyshev's results: The
best degree-$(d-1)$ polynomial approximation to $x^d$ over the
interval $[-1,1]$ achieves an error of exactly $2^{-d+1},$ see Theorem \ref{thm:chebyshev}. However, finding
the best degree-$d$ polynomial for other functions is usually  intractable, and will not be the focus in the survey. Rather, we either find a $\delta$-approximation for a suitably small $\delta$ or prove that there are none.  

\vspace{2mm}
Often, an effective way  to study $\delta$-approximations is to consider
\emph{relaxations} of the problem of finding the best uniform approximation. A natural relaxation to consider is to find the
degree-$d$ polynomial $p$ that minimizes the $\ell_2$ error
$\int_{-1}^1 (f(x)-p(x))^2 \dif{x}$. Algorithmically, we know how to solve this problem 
efficiently: It suffices to have an \emph{orthonormal} basis of
degree-$d$ polynomials $p_0(x),\ldots,p_d(x),$ \emph{i.e.}, polynomials that satisfy 
$\int_{-1}^1 p_i(x)p_j(x) \dif{x} = 0$ if $i \neq j$ and $1$
otherwise. Such an orthonormal basis can be constructed by applying
Gram-Schmidt orthonormalization to the polynomials $1,x,\ldots,x^d$
with respect to the uniform measure on $[-1,1]$ \footnote{These orthogonal polynomials are given explicitly by
  $\left\{\sqrt{\nfrac{(2d+1)}{2}}\cdot L_d(x)\right\}$, where $L_d(x)$ denotes
  the degree-$d$ Legendre polynomials. See~\cite{Szego24}.}. Given such
an orthonormal basis, the best approximation is given by $p(x) =
\sum_i \hat{f}_ip_i(x),$ where $\hat{f}_i = \int_{-1}^1
f(x)p_i(x)\dif{x}.$

\vspace{2mm}
Given a relaxation, we must consider how good that relaxation is,
\emph{i.e.}, if $p(x)$ is the best $\ell_2$-approximation to the
function $f(x),$ how does it compare to the best uniform approximation
to $f(x)$? For the straightforward relaxation above, the approximation turns out to not be meaningful. However, if we modify the relaxation to minimize the $\ell_2$ error with
respect to the weight function $w(x)\defeq \nfrac{1}{\sqrt{1-x^2}},$ \emph{i.e.},
minimize $\int_{-1}^1 (f(x)-p(x))^2 \frac{\dif{x}}{\sqrt{1-x^2}},$ then when $f$ is continuous the best degree-$d$ $\ell_2$-approximation with respect to $w$ turns out be an $O(\log d)$ approximation for the best uniform approximation (see~\cite[Section 2.4]{Rivlin69} for a
proof). Formally, if we let $p$ be the degree-$d$ polynomial that
minimizes $\ell_2$-approximation with respect to $w,$ and let $p^\star$ be the
best degree-$d$ uniform approximation, then 
\begin{equation}
\label{eq:log-d-approx}
\sup_{x \in [-1,1]} |f(x)-p(x)| \le O(\log d)\cdot \sup_{x \in
  [-1,1]} |f(x)-p^\star(x)|.
\end{equation}

\noindent
The orthogonal polynomials obtained by applying the Gram-Schmidt process with weight $w$ turn out to be  \emph{Chebyshev Polynomials}, which are central to approximation theory due to their important extremal properties.

\subsection{Chebyshev Polynomials}
There are several ways to define Chebyshev polynomials. For a non-negative integer $d$, if $T_d(x)$ denotes the Chebyshev
polynomial of degree $d$, then it can be defined recursively as follows: $T_0(x)
\defeq 1, T_1(x) \defeq x,$ and for $d \ge 2,$
\begin{equation}
\label{eq:prelims:chebyshev-recurrence}
T_{d}(x) \defeq 2xT_{d-1}(x) - T_{d-2}(x).
\end{equation}

\noindent
For convenience, we extend the definition of
Chebyshev polynomials to negative integers by defining $T_d(x) \defeq T_{|d|}(x)$ for $d<0$. It is easy to verify that with this definition, the recurrence given by
\eqref{eq:prelims:chebyshev-recurrence} is satisfied for
all integers $d.$
Rearranging ~\eqref{eq:prelims:chebyshev-recurrence}, we
obtain the following:
\begin{proposition}
\label{prop:prelims:chebyshev-relation}
The Chebyshev polynomials $\{T_d\}_{d \in \bz}$ satisfy the following
relation for all $d \in \bz,$ 
\[xT_d(x) = \frac{T_{d+1}(x) + T_{d-1}(x)}{2}.\]
\end{proposition}

\noindent An important property of  Chebyshev polynomials, which is often used to define them, is given
by the following proposition which asserts that the Chebyshev polynomial of degree $d$ is  exactly the polynomial that arises when one writes $\cos (d\theta)$ as a polynomial in $\cos \theta.$
\begin{proposition}
\label{prop:chebyshev:cosine}
For any $\theta \in \rea,$ and any integer $d,$ $T_d(\cos \theta)
=\cos (d \theta)$. 
\end{proposition}
\noindent
This can be easily verified as follows. First,
note that $T_0(\theta) = \cos(0) = 1$ and $T_1(\theta) = \cos(\theta)
= x$. Additionally, $\cos (d\theta) = 2 \cdot \cos \theta \cdot \cos
((d-1)\theta) - \cos ((d-2)\theta)$ and, hence, the recursive
definition for Chebyshev polynomials applies.
This proposition also immediately implies that over the interval
$[-1,1],$ the value of any Chebyshev polynomials is bounded by $1$ in
magnitude.
\begin{proposition}
\label{prop:poly:prelims:chebyshev-bdd}
For any integer $d,$ and $x \in [-1,1],$ we have $|T_d(x)| \le 1.$
\end{proposition}

\noindent
In fact, Proposition~\ref{prop:chebyshev:cosine} implies that, over
the interval $[-1,1],$ the polynomial $T_d(x)$ achieves its extremal
magnitude at exactly $d+1$ points $x=\cos (\nfrac{j\pi}{d}),$ for
$j=0,\ldots,d,$ and the sign of $T_d(x)$ alternates at these
points. This is ideally suited for an application of Theorem~\ref{lem:poussin}, and we can now
prove Chebyshev's result mentioned in the previous section.
\begin{theorem}\label{thm:chebyshev}
  For every positive integer $d,$ the best degree-$(d-1)$ polynomial
  approximation to $x^d$ over $[-1,1],$ achieves an approximation
  error of $2^{-d+1},$ \emph{i.e.},
$\inf_{p_{d-1} \in \Sigma_{d-1}}\sup_{x \in [-1,1]} |x^d - p_{d-1}(x)| = 2^{-d+1}.$
\end{theorem}
\begin{proof}
  Observe that the leading coefficient of $T_d(x)$ is $2^{d-1}$ and,
  hence, $2^{-d+1}T_d(x)-x^d$ is a polynomial of degree $(d-1).$ The
  error this polynomial achieves in approximating $x^d$ on $[-1,1]$ is
  $2^{-d+1}T_d(x),$ which is bounded in magnitude on $[-1,1]$ by
  $2^{-d+1},$ and achieves the value $\pm 2^{-d+1}$ at $d+1$ distinct
  points with alternating signs. The result now follows from
  Theorem~\ref{lem:poussin}.
\end{proof}

\noindent
The fact that $T_d(x)$ takes alternating $\pm 1$ values $d+1$ times in
$[-1,1],$ leads to another property of the Chebyshev polynomials: 
\begin{proposition}
\label{prop:chebyshev:extremal-growth}
  For any degree-$d$ polynomial $p(x)$ such that $|p(x)| \le 1$
  for all $x \in [-1,1],$ for any $y$ such that $|y| > 1,$ we have
  $|p(y)| \le |T_d(y)|.$
\end{proposition}
\noindent
\begin{proof}
For sake of contradiction, let $y$ be such  that $|p(y)|> |T_d(y)|$ and let $q(x) \defeq
\frac{T_d(y)}{p(y)}\cdot p(x).$ Hence, $\left|q(x)\right| < |p(x)| \le 1$
for all $x \in [-1,1],$ and $q(y)=T_d(y).$ Thus, strictly between any two consecutive points where  $T_d(x)$ alternates between $+1$ and $-1$, there must  be a point $x_i$ at which  $T_d(x_i) = q(x_i)$ since $|q(x)|<1$ in $[-1,1].$ Hence, $T_d(x)-q(x)$ has at least $d$ distinct zeros
in the interval $[-1,1],$ and another zero at $y.$ 
Hence it is a non-zero polynomial of degree at most $d$ with $d+1$
roots, which is a contradiction.
\end{proof}

\noindent
This proposition is used to prove a lower bound for rational
approximations to $e^{-x}$ in Section~\ref{sec:rational:vanilla}.
In order to do so, we need to upper bound their growth. This can be
achieved using the following closed-form expression for $T_d(x)$ which can be easily verified using the recursive definition
of Chebyshev polynomials.

\begin{proposition}
\label{prop:chebyshev:closed-form}
For any integer $d,$ and $x,$ we have
\[T_d(x) = \frac{1}{2} \left( x+\sqrt{x^2-1} \right)^d + \frac{1}{2}\left( x
    - \sqrt{x^2-1} \right)^d.\]
\end{proposition}

\section{Polynomial Approximations}\label{sec:poly-apx}
In this section, we use Chebyshev polynomials and their
properties  to construct polynomial approximations to
some fundamental functions such as the monomial $x^s,$ and the
exponential function $e^{-x}.$ We also introduce the famous Markov's theorem, and prove a special case, which is then  used to prove lower bounds on the degree of best polynomial approximations. 

\subsection{Approximating $x^s$ on $[-1,1]$}
\label{sec:xk-approx}
Recall from Proposition \ref{prop:prelims:chebyshev-relation} that for any $d,$ we can write $xT_d(x) =
\nfrac{1}{2} \cdot (T_{d-1}(x) + T_{d+1}(x)).$
 If we let $Y$ be a
random variable that takes values $1$ and $-1,$ with probability
$\nfrac{1}{2}$ each, we can write $xT_d(x) = \av_Y[T_{d+Y}(x)].$ This
simple observation can be iterated to obtain an expansion of the
monomial $x^s$ for any positive integer $s$ in terms of the
Chebyshev polynomials.
Throughout this section, let $Y_1,Y_2,\ldots$ be i.i.d. variables
taking values $1$ and $-1$ each with probability $\nfrac{1}{2}$. For any
integer $s \ge 0,$ define the random variable $D_s \defeq \sum_{i=1}^s
Y_i$ where $D_0 \defeq 0.$
\begin{lemma}
\label{lem:xk-exact}
For any integer $s \ge 0,$ we have,
$\av_{Y_1,\ldots,Y_s} [T_{D_s}(x)] = x^s.$
\end{lemma}
\begin{proof}
We proceed by induction. For $s=0,$ $D_s = 0$ and, hence,
$\av[T_{D_s}(x)] = T_0(x) = 1 = x^0.$ Moreover, for any $s \ge 0,$
\begin{eqnarray*}
x^{s+1} \stackrel{\rm Induction}{=} x \cdot
\av_{Y_1,\ldots,Y_s}{T_{D_s}(x)} 
 & = & \av_{Y_1,\ldots,Y_s}[x \cdot T_{D_s}(x)]\\
& \stackrel{\rm Prop. \; \ref{prop:prelims:chebyshev-relation}}{=} &
\av_{Y_1,\ldots,Y_s}\left[\frac{T_{D_s+1}(x) +T_{D_s-1}(x)}{2}\right] \\
& =  &\av_{Y_1,\ldots,Y_s,Y_{s+1}}[T_{D_{s+1}}(x)].
\end{eqnarray*}
\end{proof}

\noindent
Lemma~\ref{lem:xk-exact} allows us to obtain polynomials that
approximate $x^s,$ but have degree close to  $\sqrt{s}.$ The main
observation is that the probability that $|D_s| \ge d$ is small. In particular, using Chernoff bounds, the probability that $|D_s| >
\sqrt{2s \log \nfrac{2}{\delta}} \defeq \widehat{d}$ is at most
$\delta.$ Moreover, since $|T_{D_s}(x)| \le 1$ for all $x \in [-1,1],$
we can ignore all terms with degree greater than $\widehat{d}$ without
incurring an error greater than $\delta.$ We now prove this formally.

Let $\mathbbm{1}_{|D_s| \le d}$ denote the indicator variable for the
event that $|D_s| \le d.$ Our polynomial of degree $d$ approximating
$x^s$ is obtained by truncating the above expansion to degree $d,$ \emph{i.e.},
\begin{equation}
\label{eq:approx-poly}
  p_{s,d}(x) \defeq \av_{Y_1,\ldots,Y_s}\left[ T_{D_s}(x) \cdot
    \mathbbm{1}_{|D_s| \le d} \right].
\end{equation}
Since $T_{D_s}(x)$ is a polynomial of degree $|D_s|,$ and the
indicator variable $\mathbbm{1}_{|D_s| \le d}$ is zero whenever $|D_s|
> d,$ we obtain that $p_{s,d}$ is a polynomial of degree at most $d.$
\begin{theorem}
\label{thm:xk-approx}
For any positive integers $s,d,$ the degree-$d$ polynomial $p_{s,d}$
satisfies
\[\sup_{x \in [-1,1]} |p_{s,d}(x) - x^s| \le 2e^{-\nfrac{d^2}{2s}}.\]
Hence, for any $\delta >0,$ and $d \ge \ceil{\sqrt{2s \log
    \nfrac{2}{\delta}}},$ we have $\sup_{x \in [-1,1]} |p_{s,d}(x) -
x^s| \le \delta.$
\end{theorem}
\begin{proof}
Using Chernoff bounds (see \cite[Chapter 4]{Mitzenmacher2005}), we know that 
\[\av_{Y_1,\ldots,Y_s}\left[  \mathbbm{1}_{|D_s| > d} \right]  = \pr_{Y_1,\ldots,Y_s} [|D_s| > d] = \pr_{Y_1,\ldots,Y_s} \left[\left|
    \sum_{i=1}^s Y_i \right| > d \right] \le 2e^{-\nfrac{d^2}{2s}}. \]
Now, we can bound the error in approximating $x^s$ using $p_{s,d}.$
\begin{eqnarray*}
\sup_{x \in [-1,1]}  |p_{s,d}(x) - x^s| &
\stackrel{{\rm Lem}.~\ref{lem:xk-exact}} {=} & \sup_{x \in [-1,1]}
\left| \av_{Y_1,\ldots,Y_s}\left[ T_{D_s}(x) \cdot
    \mathbbm{1}_{|D_s| > d} \right] \right| \\
& \le & \sup_{x \in [-1,1]}
\av_{Y_1,\ldots,Y_s}\left[ \left| T_{D_s}(x) \right| \cdot
    \mathbbm{1}_{|D_s| > d} \right] \\
& \le & 
\av_{Y_1,\ldots,Y_s}\left[  \mathbbm{1}_{|D_s| > d} \cdot 
 \sup_{x \in [-1,1]}\left| T_{D_s}(x) \right| \right] 
 \stackrel{{\rm Prop.} ~\ref{prop:poly:prelims:chebyshev-bdd}}{\le} 
 \av_{Y_1,\ldots,Y_s}\left[  \mathbbm{1}_{|D_s| > d} \right] \le 2e^{-\nfrac{d^2}{2s}},
\end{eqnarray*}
which is smaller than $\delta$ for $d \ge \ceil{\sqrt{2s \log
    \nfrac{2}{\delta}}}$.
\end{proof}

\noindent
Over the next several sections, we explore several interesting
consequences of this seemingly simple approximation. In Section~\ref{sec:exp-poly-ub}, we use this approximation to give improved polynomial
approximations to the exponential function. In
Sections~\ref{sec:applications:cg} and \ref{sec:lanczos}, we use it to give  fast
algorithms for solving linear systems and computing eigenvalues. We prove that the
$\sqrt{s}$ dependence is optimal in Section~\ref{sec:lb}.

\subsection{Approximating $e^{-x}$ on $[0,b]$}
\label{sec:exp-poly-ub}
In this section we consider the problem of approximating $e^{-x}$ over the interval $[0,\infty).$ 
The first problem one faces when one looks for polynomial approximations for $e^{-x}$ over $[0,\infty)$ is that none exist. The reason is that a polynomial goes to infinity with $x$, while $e^{-x}$ goes to $0.$ However, if one restricts to approximating $e^{-x}$ over an interval $[0,b],$ then a simple approach is to
truncate the Taylor series expansion of $e^{-x}.$ It is easy to show
that using roughly $b+\log \nfrac{1}{\delta}$ terms in the
expansion suffices to obtain a $\delta$ approximation. 
We  show
that the approximation for $x^s$ we developed in the previous section
allows us to obtain a quadratic improvement over this simple approximation.
\begin{theorem}
\label{thm:exp-poly-ub}
For every $0 < b,$ and $0<\delta \le 1$, there exists a polynomial
${r}_{b,\delta}$ that satisfies, $\sup_{x \in [0,b]}
|e^{-x}-{r}_{b,\delta}(x)| \le \delta,$ and has degree
$O\left(\sqrt{\max\{b,\log \nfrac{1}{\delta} \} \cdot \log
    \nfrac{1}{\delta}} \right)$.
\end{theorem}

\noindent

After a scaling and translation, it suffices to
approximate the function $e^{-\lambda x - \lambda}$ over the interval
$[-1,1],$ where $\lambda = \nfrac{b}{2}.$ As mentioned before, if we truncate
its Taylor expansion, we obtain $\sum_{i=0}^t e^{-\lambda}
\frac{(-\lambda)^i}{i!}x^i$ as a candidate approximating polynomial.
Our candidate polynomial is obtained by a general strategy that approximates each monomial
$x^i$ in this truncated series by the polynomial $p_{i,d}$ from the
previous section. Formally, 
\begin{equation*}
q_{\lambda,t,d}(x) \defeq \sum_{i=0}^t e^{-\lambda}\tfrac{(-\lambda)^i}{i!} p_{i,d}(x).
\end{equation*}
Since $p_{i,d}(x)$ is a polynomial of degree at most $d,$ the polynomial $q_{\lambda,t,d}(x)$
is also of degree at most $d.$ We now prove that for $d$
roughly $\sqrt{\lambda},$ the polynomial $q_{\lambda,t,d}(x)$ gives a good approximation
to $e^{-\lambda x}$ (for an appropriate choice of $t$).  
\begin{lemma}
\label{lem:poly:up-bd}
For every $\lambda > 0$ and $\delta \in (0,\nfrac{1}{2}],$ we can
choose $t = O(\max\{\lambda,\log \nfrac{1}{\delta}\}),$ and $d =
O\left(\sqrt{t \log
    \nfrac{1}{\delta}}\right)$ such that the polynomial
$q_{\lambda,t,d}$ defined above, $\delta$-approximates the function
$e^{- \lambda - \lambda x}$ over the interval $[-1,1],$ \emph{i.e.},
\[\sup_{x \in [-1,1]} \left|e^{-\lambda - \lambda x} - q_{\lambda,t,d}(x)
\right| \le \delta.\]
\end{lemma}
\begin{proof}
We first expand the function $e^{-\lambda - \lambda x}$ via its Taylor
series expansion around $0,$ and then split it into two parts, one
containing terms with degree at most $t,$ and the remainder.
\begin{eqnarray*}
  \sup_{x \in [-1,1]} \left|e^{-\lambda - \lambda x} - q_{\lambda,t,d}(x)
  \right| 
  & \le & \sup_{x \in [-1,1]} \left|  \sum_{i=0}^t e^{-\lambda} \tfrac{(-\lambda)^i}{i!}
    (x^{i} - p_{i,d}(x)) \right| +
  \sup_{x \in [-1,1]} \left|  \sum_{i=t+1}^\infty e^{-\lambda}
    \tfrac{(-\lambda)^i}{i!}   x^{i}\right| \\
  & \le & \sum_{i=0}^t e^{-\lambda} \tfrac{\lambda^i}{i!}
  \sup_{x \in [-1,1]} \left|x^{i} - p_{i,d}(x) \right| + e^{-\lambda} \sum_{i=t+1}^\infty 
  \tfrac{\lambda^i}{i!}.
\end{eqnarray*}
From Theorem~\ref{thm:xk-approx}, we know that $p_{i,d}$ is a good
approximation to $x^i,$ and we can use it to bound the first error
term.
\begin{align*}
  \sum_{i=0}^t e^{-\lambda} \tfrac{\lambda^i}{i!}  \sup_{x \in [-1,1]}
  \left|x^{i} - p_{i,d}(x) \right| & \le \sum_{i=0}^t e^{-\lambda}
  \tfrac{\lambda^i}{i!} \cdot 2e^{-\nfrac{d^2}{2i}} \\
  & \le 2e^{-\nfrac{d^2}{2t}} \cdot \sum_{i=0}^\infty e^{-\lambda}
  \tfrac{\lambda^i}{i!}  = 2e^{-\nfrac{d^2}{2t}}.
\end{align*}
For the second term, we use the lower bound $i! \ge
\left(\nfrac{i}{e}\right)^i,$ and assume $t \ge \lambda e^2$ to
obtain
\[ e^{-\lambda} \sum_{i=t+1}^\infty \tfrac{\lambda^i}{i!} \le
e^{-\lambda} \sum_{i=t+1}^\infty \left(\tfrac{\lambda e}{i} \right)^i
\le e^{-\lambda} \sum_{i=t+1}^\infty e^{-i} \le e^{-\lambda -t}.\]
Thus, if we let $t = \ceil{\max\{\lambda e^2, \log \nfrac{2}{\delta}\}}$ and
$d = \ceil{\sqrt{2t\log \nfrac{4}{\delta}}},$ combining the above and using
$\lambda > 0,$ we obtain 
$ \sup_{x \in [-1,1]} \left|e^{-\lambda - \lambda x} - q_{\lambda,t,d}(x)
\right| \le 2e^{-\nfrac{d^2}{2t}} + e^{-\lambda -t} \le
\nfrac{\delta}{2} + \nfrac{\delta}{2} \le \delta.$
\end{proof}

\noindent
Now, we can  complete the proof of
Theorem~\ref{thm:exp-poly-ub}.
\begin{proof}[Proof of Theorem \ref{thm:exp-poly-ub}]
  Let $\lambda \defeq \nfrac{b}{2},$ and let $t$ and $d$ be given by
  Lemma~\ref{lem:poly:up-bd} for the given value of $\delta$.
 Define $r_{b,\delta} \defeq 
q_{\lambda,t,d}\left( \nfrac{1}{\lambda} \cdot
  \left(x-\nfrac{b}{2}\right) \right),$ where $q_{\lambda,t,d}$ is the
polynomial given by Lemma~\ref{lem:poly:up-bd}. Then, 
\begin{eqnarray*}
\sup_{x \in [0,b]} \left| e^{-x} - r_{b,\delta}(x) \right| =
\sup_{x \in [0,b]} \left| e^{-x} -
q_{\lambda,t,d}\left( \nfrac{1}{\lambda} \cdot
\left(x-\nfrac{b}{2}\right) \right) \right| 
 = \sup_{z \in [-1,1]} \left| e^{-\lambda z -
\lambda} - q_{\lambda,t,d}\left(z \right) \right| \le \delta,
\end{eqnarray*}
where the last inequality follows from the guarantee of
Lemma~\ref{lem:poly:up-bd}.  The degree of $r_{b,\delta}(x)$ is the
same as that of $q_{\lambda,t,d}(x),$ \emph{i.e.}, $d =
O\left(\sqrt{\max\{b,\log \nfrac{1}{\delta}\} \cdot \log
    \nfrac{1}{\delta}} \right)$.
\end{proof}

\noindent
Theorem \ref{thm:exp-poly-ub} is implicit in the work by Hochbruck and
Lubich~\cite{HochbruckL97}. A weaker version of this theorem has also been proved
in~\cite{OrecchiaSV12}.

\subsection{Markov's Theorem and Lower Bounds for Polynomial Approximations}
\label{sec:lb}
In this section, we  prove that the bounds in the previous
section are essentially optimal. Specifically, we show that the
degrees of the approximating polynomials require a $\sqrt{s}$
dependence for approximating $x^s$ on $[-1,1],$ and a $\sqrt{b}$
dependence for approximating $e^{-x}$ on $[0,b].$ Such lower bounds
often use the following well known Markov's theorem from
approximation theory.
\begin{theorem}[Markov's Theorem, see \cite{Cheney66}]
\label{thm:markov}
Let $p$ be a degree-$d$ polynomial such that $|p(x)| \le 1$ for any $x
\in [-1,1].$ Then, for all $x \in [-1,1],$ the derivative of $p,$ $p^{(1)}$ 
satisfies $|p^{(1)}(x)| \le d^2.$
\end{theorem}

\noindent
In fact, the above theorem is another example of an extremal property
of the Chebyshev polynomials since they can be seen to be a tight
example for this theorem. The above theorem also generalizes to higher
derivatives, where it implies that for any such $p,$ we have
$|p^{(k)}(x)| \le \sup_{y \in [-1,1]} |T^{(k)}_d(y)|,$ for any $k$ and
$x \in [-1,1]$~\cite[Section 1.2]{Rivlin69}. This was proved by V. A. Markov ~\cite{MarkovVA1892}.

\vspace{2mm}
We sketch a proof of the following special case of
Markov's theorem, based on the work by Bun and
Thaler~\cite{BunT13}, that  suffices for proving our lower bounds.
\begin{lemma}
\label{lem:markov-sp-case}
For any degree-$d$ polynomial $q$ such that $|q(x)| \le 1$ for all $x
\in [-1,1],$  $|q^{(1)}(1)| \le d^2.$
\end{lemma}

\noindent
We now present the main idea behind both of the lower bound proofs. Say
$p(x)$ is the approximating polynomial. We start by using the bound on
the approximation error to bound the range of values taken by the
polynomial in the interval. The crux of both the proofs  is to
show that there exists a point $t$ in the approximation interval such
that $|p^{(1)}(t)|$ is large. Once we have such a lower bound on the
derivative of $p$, a lower bound on the degree of $p$ follows by
applying the above lemma to a polynomial $q$ obtained by a linear
transformation of the input variable that maps $p(t)$ to $q(1).$ In
order to show the existence of a point with a large derivative, we use
the Mean Value theorem and the fact that our polynomial is a good approximation
to the function of interest.  We now use this strategy to show that
any polynomial that approximates $e^{-x}$ on $[0,b]$ to within
$\nfrac{1}{8}$ must have degree at least $\nfrac{\sqrt{b}}{3}.$

\begin{theorem}
\label{lem:exp-poly-lb}
For every $b \ge 5,$ and $\delta \in (0,\nfrac{1}{8}],$ any
polynomial $p(x)$ that approximates $e^{-x}$ uniformly over the
interval $[0,b]$ up to an error of $\delta,$ must have degree at least
$\nfrac{1}{3}\cdot\sqrt{b}\ .$
\end{theorem}
\begin{proof}
Suppose $p$ is a degree-$d$ polynomial that is a uniform approximation
to $e^{-x}$ over the interval $[0,b]$ up to an error of
$\delta$. Thus, for all $x \in [0,b],$ we have $e^{-x} - \delta \le p(x) \le
e^{-x} + \delta.$ Hence, $\sup_{x \in [0,b]} p(x) \le 1 +\delta$ and
$\inf_{x \in [0,b]} p(x) \ge - \delta.$

\vspace{2mm}
Assume that $\delta \le \nfrac{1}{8},$ and $b \ge 5 > 3\log_e 4.$ Applying
the Mean Value theorem (see \cite[Chapter 5]{rudin-principles}) on the interval $[0,\log_e 4],$ we know
that there exists $t \in [0,\log_e 4],$ such that 
\begin{align*}
|p^{(1)}(t)| = \left\lvert \frac{p(\log_e 4) - p(0)}{\log_e 4} \right
\rvert &\ge \frac{(1 -\delta) - (e^{ -\log_e 4} + \delta) }{\log_e 4}
\ge \frac{1}{2\log_e 4}
\end{align*}

\noindent
Consider the polynomial $q(x) \defeq \frac{1}{1+2\delta} \left(2
  p\left(\frac{t(1+x) + b(1-x)}{2}\right)-1\right).$ Since $p([0,b])
\subseteq [-\delta,1+\delta],$ we obtain $|q(x)| \le 1$ for all $x \in
[-1,1].$ Thus, using Lemma~\ref{lem:markov-sp-case}, we obtain $|q^{(1)}(1)|
\le d^2.$ This implies that  $|q^{(1)}(1)|=(b-t)|p^{(1)}(t)| \le d^2(1+2\delta).$ Plugging in
the lower bound on $|p^{(1)}(t)|$ proved above and rearranging, we obtain
$ d \ge \sqrt{\frac{b-t}{2\cdot \nfrac{5}{4}\cdot\log_e 4} } \ge
\frac{1}{3}\cdot\sqrt{b},$ where the last step uses $t \le \log_e 4
\le \nfrac{b}{3}.$
\end{proof}

\noindent
A similar proof strategy shows the tightness of the $\sqrt{s}$ bound for
approximating $x^s$ on the interval $[-1,1].$ In this case, we show
that there exists a $t \in [1-\nfrac{1}{s},1]$ such that $|p^{(1)}(t)| \ge
\Omega(s)$ (assuming $\delta$ small enough).  The lower bound now
follows immediately by applying Lemma~\ref{lem:markov-sp-case} to the
polynomial $\nfrac{1}{1+\delta} \cdot p(tx).$
Now, we give a proof of the special case of Markov's theorem
given by Lemma~\ref{lem:markov-sp-case}.
\begin{proof}[Proof of Lemma~\ref{lem:markov-sp-case}]
If we expand the polynomial $q$ around $x=1$ as follows, $q(x) = c_0 + c_1
(x-1) + \ldots + c_d (x-1)^d,$ we have $q^{(1)}(1) = c_1.$ Hence, we can
express the upper bound on $q^{(1)}(1)$ as the optimum of the following
linear program where the $c_i$s are variables and there are an infinite number of
constraints: 
$$\max c_1 \ \  {\rm s.t.} \ \ \  \left| \sum_{i=0}^d c_i(x-1)^i
\right| \le 1 \ \ \ \forall x \in [-1,1].$$ Since $(-c_i)_i$ is a feasible
solution whenever $(c_i)_i$ is, it suffices to maximize $c_1$ instead
of $|c_1|.$

\vspace{2mm}
Now, we relax this linear program, and drop all constraints except
for $x = \cos(\nfrac{k\pi}{d})$ for integral $k$ between $0$ and $d:$ \emph{i.e.},
$\max c_1$ subject to $ \sum_{i=0}^d c_i(x-1)^i \le 1$ for $x =
\cos(\nfrac{k\pi}{d})$ with even $k,$ and $\sum_{i=0}^d c_i(x-1)^i \ge
-1$ for $x = \cos(\nfrac{k\pi}{d})$ with odd $k.$\footnote{Though
  these particular values seem magical, they are exactly the extremal
  points of the Chebyshev polynomial $T_d(x),$ which is known to be a
  tight example for Markov's theorem.} It suffices to show
that the optimum of this linear program is bounded above by $d^2.$
We  show this by constructing a feasible solution to its dual
program.
We can write the dual to the restricted linear program as follows: 
$$\min
\sum_{i=0}^d y_i \ \ {\rm s.t.} \ \  Ay = e_1 \ \  {\rm and} \ \  y_j \ge 0 \ \   \forall j.$$
Here $e_1 \in \rea^{d+1}$ is the vector $(0,1,0,\ldots,0)^\top,$ and
$A$ is the matrix defined by $A_{ij} \defeq (-1)^j (\cos (\nfrac{j\pi}{d})-1)^i,$
where $i=0,\ldots,d,$ and $j=0,\ldots,d.$
Using  elementary trigonometric identities~(see~\cite{BunT13}), one
can show that
$$\textstyle y =
\left(\frac{2d^2+1}{6}, \csc^2 \frac{\pi}{2d},\csc^2 \frac{\pi}{d},
  \ldots, \csc^2 \frac{(d-1)\pi}{2d}, \frac{1}{2}\right)^\top$$ is, in
fact, the unique solution to $Ay=e_1,$ and satisfies $\sum y_i = d^2.$
It trivially satisfies the positivity constraints and, hence, by weak
duality implies an upper bound of $d^2$ on the optimum value of primal
linear program.
\end{proof}

\section{Rational Approximations}\label{sec:rational-apx}

In this section we introduce approximations to functions by rational functions such as $\frac{p(x)}{q(x)}$ where $p,q$ are polynomials. The error in approximation is again measured as the worst error in the interval of interest and we would be interested in trade-off  between the error and the maximum of the degrees of $p,q.$
The surprising power of rational approximations was first demonstrated by  Newman~\cite{Newman64} who showed  that rational
approximations can be much more powerful than polynomial
approximations. He proved that the best degree-$d$ rational
approximation to the function $|x|$ on $[-1,1]$ achieves an
approximation error of $e^{-\Theta(\sqrt{d})}.$ Contrast this with the fact that  the best degree-$d$ polynomial approximation to $|x|$
on $[-1,1]$ only achieves an error of $\Theta(\nfrac{1}{d}).$

\vspace{2mm}
Unlike the lower bound results proved in the previous section, we show that rational functions can provide
approximations to $e^{-x}$ that hold for all $x \geq 0,$ and achieve an
approximation error that is exponentially small in their degree. We also show how to construct such rational approximations which in addition have negative poles. Such rational functions are extremely useful in applications, see Section \ref{sec:applications:matrix-exp}.

\subsection{Approximating $e^{-x}$ on $[0,\infty)$}
\label{sec:rational:vanilla}
In this section we  show that, somewhat surprisingly,  there exist simple rational
functions of the form $\nfrac{1}{p(x)},$ where $p$ is a low degree
polynomial, that approximate $e^{-x}$ over $[0,\infty),$ up to an approximation error that decays
exponentially with the degree of the approximation. We  also show
that no rational approximation of the form $\nfrac{1}{p(x)}$ can do
much better.

\subsubsection*{Upper Bound}
In the last section, we showed that the partial sums of the Taylor
series expansion of $e^{-x}$ requires a large degree in order to
provide a good approximation over a large interval. We now show that if we instead truncate the Taylor series
expansion of $e^x = \nfrac{1}{e^{-x}}$ to degree $d,$ and take its
reciprocal, we can approximate $e^{-x}$ on $[0,\infty)$ up to
$2^{-\Omega(d)}$ error.
We let $S_d(x) \defeq \sum_{k=0}^d
\frac{x^k}{k!}.$ 
\begin{theorem} 
\label{thm:exp-rational}
For all integers $d \ge 0,$
\[\sup_{x \in [0,\infty)} \left| \frac{1}{S_d(x)} - e^{-x}\right| \le
2^{-\Omega(d)}.\]
Hence, for any $\delta > 0,$ we have a rational function of degree
$O(\log \nfrac{1}{\delta})$ that is a $\delta$-approximation to
$e^{-x}$. 
\end{theorem}

\begin{proof}
First, observe that for all $d,$ and all $x \in \left[0,\infty\right),$ we have
$S_d(x) \le e^{x}$ and, hence, $\nfrac{1}{S_d(x)} - e^{-x} \ge 0.$
We  divide $\left[0,\infty\right)$ into three intervals:
$\left[0,\nfrac{d+1}{3}\right), \left[\nfrac{d+1}{3},\nfrac{2(d+1)}{3}\right),$ and
$\left[\nfrac{2(d+1)}{3},\infty\right),$ and show a bound on the approximation
error on each of these intervals.
If $x \ge \nfrac{2(d+1)}{3},$ both the terms are exponentially
small. Using $S_d(x) \ge \nfrac{x^d}{d!}$ and $d! \le
  (\frac{d+1}{2})^d,$ we obtain 
\begin{align}
  {\textstyle \forall x \in \left[0, \frac{d+1}{3}  \right)} 
\qquad & \left| \frac{1}{S_d(x)} - e^{-x}\right| \le \frac{1}{S_d(x)} \le
\frac{d!}{x^d} \le \left(\frac{d+1}{2x} \right)^d \le
\left(\frac{3}{4} \right)^d = 2^{-\Omega(d)}, \nonumber \\
\intertext{Now, assume that $x < \nfrac{2(d+1)}{3}.$ We have,}
\label{eq:exp-rational}
 \left| \frac{1}{S_d(x)} - e^{-x}\right| & =\frac{e^{-x}}{S_d(x)}
\left(\frac{x^{d+1}}{(d+1)!} + \frac{x^{d+2}}{(d+2)!} + \ldots \right)
\nonumber \\
& \le \frac{e^{-x}}{S_d(x)} \cdot \frac{x^{d+1}}{ (d+1)!}\left( 1+
  \frac{x}{d+1} + \frac{x^2}{(d+1)^2} + \ldots \right) \le 3
\frac{e^{-x}}{S_d(x)} \cdot \frac{x^{d+1}}{ (d+1)!}.
\intertext{If $x \in \left[ \nfrac{d+1}{3},\nfrac{2(d+1)}{3} \right),$ we use that
$e^{-x}$ is exponentially small, and show that the numerator is not
much larger than $S_d(x).$ We use $S_d(x) \ge \nfrac{x^d}{d!}$ in
\eqref{eq:exp-rational} to obtain}
 {\textstyle \forall x \in \left[ \frac{d+1}{3},\frac{2(d+1)}{3} \right)} \qquad
 & \left| \frac{1}{S_d(x)} - e^{-x}\right|  \le 3 e^{-\frac{d+1}{3}}
\cdot \frac{x}{d+1} \le 2 e^{-d/3} = 2^{-\Omega(d)}. \nonumber
\intertext{Finally, if $x < \frac{d+1}{3},$ we use that $S_d(x)$ is an
exponentially good approximation of $e^x$ in this range. Using $(d+1)!
\ge \left(\frac{d+1}{e}\right)^{d+1}$ and $S_d(x) \ge 1$ in
\eqref{eq:exp-rational} to obtain}
{\textstyle \forall x \in \left[ \frac{2(d+1)}{3},\infty \right)} \qquad
 & \left| \frac{1}{S_d(x)} - e^{-x}\right|  \le 3
\left(\frac{xe}{d+1}\right)^{d+1} \le 3
\left(\frac{e}{3}\right)^{d+1} = 2^{-\Omega(d)}. \nonumber
\end{align}
\end{proof}

\noindent
A more careful argument shows that, in fact, $\nfrac{1}{S_d(x)}$ approximates
$e^{-x}$ up to an error of $2^{-d}$ (Lemma 1 in
\cite{CodyMV69}). 

\subsubsection*{Lower Bound}
We now show that polynomials other than $S_d(x)$ cannot do much
better. We give a simple proof that shows that for any rational function
of the form $\nfrac{1}{p_d(x)}$ that approximates $e^{-x}$ on
$[0,\infty),$ where $p_d(x)$ is a degree-$d$ polynomial, the error
cannot decay faster than an exponential in the degree.
\begin{theorem}
  For every degree-$d$ polynomial $p_d(x)$ with $d$ large enough,  $\sup_{x \in [0,\infty)} \left|e^{-x} - \nfrac{1}{p_d(x)}\right|
  \ge 50.$
\end{theorem}
\begin{proof}
Assume for sake of contradiction that for some large enough $d$ there exists a
degree-$d$ polynomial $p_d(x)$ such that $\nfrac{1}{p_d(x)}$
approximates $e^{-x}$ up to an error of $50^{-d}$ on $[0,\infty).$
Thus, for all $x \in [0,d],$ we have $\nfrac{1}{p_d(x)} \ge e^{-d} -
50^{-d} \ge \nfrac{1}{2}\cdot e^{-d},$ \emph{i.e.}, $|p_d(x)| \le
2e^d.$ Hence, the degree-$d$ polynomial $\nfrac{1}{2} \cdot {e^{-d}}
\cdot p_d\left(\nfrac{d}{2}+ \nfrac{d}{2}\cdot y\right)$ is bounded by
1 in absolute value over the interval $[-1,1].$ Using
Proposition~\ref{prop:chebyshev:extremal-growth}, which implies that the
Chebyshev polynomials have the fastest growth amongst such
polynomials, we obtain $\nfrac{1}{2}\cdot e^{-d} \cdot
p_d\left(\nfrac{d}{2}+ \nfrac{d}{2} \cdot y\right) \le T_d(y).$ Using
the closed-form expression for Chebyshev polynomials given in 
Proposition~\ref{prop:chebyshev:closed-form}, we have $T_d(y) =
\frac{1}{2}\left( \left( y+\sqrt{y^2-1} \right)^d + \left( y -
    \sqrt{y^2-1} \right)^d\right).$ For $y=7,$ we have $p_d(4d) \le
2e^d\cdot T_d(7) \le 2e^d \cdot 14^d.$ This implies that for $x=4d,$
we obtain $|e^{-x} - \nfrac{1}{p_d(x)}| \ge \nfrac{1}{p_d(x)} - e^{-x}
\ge \frac{1}{2} (14e)^{-d} - e^{-4d},$ which is larger than
$50^{-d}$ for $d$ large enough. This contradicts the assumption that
$\nfrac{1}{p_d(x)}$ approximates $e^{-x}$ for all $x \in [0,\infty)$
up to an error of $50^{-d}.$
\end{proof}

\noindent
The exact rate of decay of the best approximation for $e^{-x}$ using rational functions was a central problem in approximation theory for more than 15 years. Cody, Meinardus, and Varga~\cite{CodyMV69} were the first
to prove a lower bound of $6^{-d+o(d)}$ for rational functions of the
form $\nfrac{1}{p_d(x)}$ where $p_d$ is a degree-$d$
polynomial. Sch\"{o}nhage~\cite{Schonhage73} proved that the best
approximation of the form $\nfrac{1}{p_d(x)}$ achieves an
approximation error of $3^{-d+o(d)}.$ Newman~\cite{Newman74} showed
that even for an arbitrary degree-$d$ rational function, \emph{i.e.},
$\nfrac{p_d(x)}{q_d(x)}$ approximating $e^{-x},$ where both $p_d(x)$
and $q_d(x)$ are polynomials of degree at most $d,$ the
approximation error cannot be smaller than $1280^{-d}.$ The question
was settled by Gonchar and Rakhmanov~\cite{GoncharR83} who finally
proved that the smallest approximation error achieved by arbitrary
degree-$d$ rational functions is $c^{-d(1+o(1))},$ where $c$ is the
solution to an equation involving elliptic integrals.

\subsection{Approximating $e^{-x}$ on $[0,\infty)$ with Negative
  Poles}
\label{sec:exp-rational2}
In this section we study the question of rational approximations to $e^{-x}$ with geometric convergence and with negative zeros.
Such rational approximations have been used, in combination with the powerful Laplacian solvers  \cite{SpielmanT04, KoutisMP11,
  KelnerOSZ13} to design  near-linear time algorithms to compute approximations $\exp(-L)v$ when $L$ is a graph Laplacian; see Section \ref{sec:applications:matrix-exp}.  

\vspace{2mm}
Unfortunately, the rational approximation $\nfrac{1}{S_d(x)}$ that we
studied in the last section does not satisfy this requirement of having all negative poles. The zeros of
$S_d(x)$ have been well studied (see~\cite{Zemyan05} for a
survey). It is fairly simple to show that $S_d(x)$ has exactly one
real zero $x_d \in [-d, -1]$ if $d$ is odd, and no real zeros if $d$
is even. It is also known that the zeros of $S_d(x)$ grow linearly in
magnitude with $d.$ In fact, it was proved by Szeg\"{o}~\cite{Szego24}
that if all the (complex) zeros of $S_d$ are scaled by $d,$ as $d$
goes to infinity they converge to a point on the curve $|ze^{1-z} | =
1$ on the complex plane.

\vspace{2mm}
How about the approximation $(1+\nfrac{x}{d})^{-d}$? Trivially,
it is a simple rational function where the denominator has only 
negative zeros, and converges to $e^{-x}$ uniformly over
$[0,\infty).$ However, the convergence rate of this approximation is slow with $d$  and it is
easy to see that the error in the approximation at $x=1$ is
$\Theta(\nfrac{1}{d}).$ Saff, Sch\"{o}nhage, and Varga~\cite{SaffSV75} 
showed that for every rational function of the form
$\nfrac{1}{p_d(x)},$ where $p_d$ is a degree-$d$ polynomial with real
roots, $\sup_{x \in [0,\infty)} |e^{-x} - \nfrac{1}{p_d(x)}| =
\Omega(\nfrac{1}{d^2}).$

\vspace{2mm}
Surprisingly, the authors in~\cite{SaffSV75} showed that if we instead
consider rational functions of the form $p_d(x)(1+\nfrac{x}{d})^{-d},$
then we can approximate $e^{-x}$ up to $O(d2^{-d})$ for some
degree-$d$ polynomial $p_d(x),$ see also \cite{Andersson81}.   
Formally, 
\cite{SaffSV75} proved the following.
\begin{theorem}\label{thm:SSV}
For every $d,$ there exists a degree-$d$ polynomial $p_d$ such that,
\[\sup_{x \in [0,\infty)} \left| e^{-x} -
  \frac{p_d(x)}{(1+\nfrac{x}{d})^d}\right| \le O(d\cdot 2^{-d}).\]
\end{theorem}
\noindent
Since we seek an approximation over an infinite interval, we first apply a variable transformation to convert the interval into a
finite one. Towards this, we can write $p_d(x)(1+\nfrac{x}{d})^{-d}$ as a
degree-$d$ polynomial in $(1 + \nfrac{x}{d})^{-1}.$ Hence, in order to
make a transformation so that the new variable varies over the 
interval $[-1,1]$, we can attempt the transformation $y = 1-
2(1+\nfrac{x}{d})^{-1}.$ Thus, we are looking for a polynomial
approximation $q_d(y)$ to the function $e^{-x} = \exp\left(-d\cdot
  \nfrac{(1+y)}{(1-y)}\right).$ Observe that $y$ now varies over the
interval $[-1,1),$ and the approximation error has remained unchanged.

\vspace{2mm}
Let $f_d(y) \defeq \exp\left(-d\cdot \nfrac{(1+y)}{(1-y)}\right),$
with $f_d(1)=0.$ We could attempt to use the Taylor series
approximations in order to approximate this function. One strategy,
that can be shown to work, is to consider the polynomial $r_d(y)$
obtained by truncating, up to degree $d,$ the Taylor series expansion
of the function $f_1(y) \defeq \exp\left(-\nfrac{(1+y)}{(1-y)}\right)$
around $y=-1,$ and to consider the degree-$d^2$ polynomial $r_d^d(y)$
as the approximating polynomial. We do not pursue this approach
here.   Instead, we
now present a simplification of  the proof from~\cite{SaffSV75}. 

\begin{proof}[Proof of Theorem \ref{thm:SSV}]
We start by    relaxing  the question of a
uniform approximation of $f_d$, to an $\ell_2$- approximation
problem. However,  the relaxation is through an
intermediate $\ell_1$ problem. Let $f_d^{(k)}$ denote the
$k^\textrm{th}$ derivative of $f_d,$ \emph{i.e.}, $f_d^{(k)}(t) \defeq
\od[k]{\ }{t} f_d(t).$ Then, the following is a simple sequence of equalities and inequalities which rely on Cauchy-Schwartz.
\begin{align}
  \inf_{q_d} \sup_{y \in [-1,1)}\left|f_d(y)-q_d(y)\right| & = 
  \inf_{r_{d-1}} \sup_{y \in [-1,1)} \left| \int_{y}^1
    (f^{(1)}_d(t)-r_{d-1}(t)) \dif{t}\right| \nonumber \\
& \le \inf_{r_{d-1}}
  \int_{-1}^1 \left| f^{(1)}_d(t)-r_{d-1}(t)\right| \dif{t} \nonumber \\
& \le \sqrt{2} \inf_{r_{d-1}}
  \sqrt{\int_{-1}^1 \left( f^{(1)}_d(t)-r_{d-1}(t)\right)^2 \dif{t}}. 
\label{eq:exp-rational:l2-min}
\end{align}
The first equality holds if we take the infimum over all
degree-$(d-1)$ polynomials $r_{d-1}.$ We know how to write an explicit
solution to the optimization problem in the last expression. We
require orthogonal polynomials on $[-1,1]$ under the constant weight
function, which are given by Legendre polynomials $L_k(t) \defeq \frac{1}{2^k\cdot k!} \od[k]{\ }{t} [(t^2-1)^k],$ and
satisfy $\int_{-1}^1 L_i(t)L_j(t) \dif {t} = \frac{2}{2i+1}$ if and only if $i=j$
and 0 otherwise (see~\cite{Szego39}). Hence, we can write the last expression explicitly to
obtain
\begin{align}
\label{eq:ssv:error}
  \inf_{q_d} \sup_{y \in (-1,1]}\left|f_d(y)-q_d(y)\right| & \le
  \sqrt{\sum_{k \ge d} (2k+1)\gamma_k^2 },
\end{align}
where $\gamma_k$ denotes the inner product with the $k^\textrm{th}$
Legendre polynomial $\gamma_k \defeq \int_{-1}^1 f^{(1)}_d(t)L_k(t) \dif{t}.$
Plugging in the definition of Legendre polynomials, and using the
integration by parts successively, we obtain 
\begin{equation}
\label{eq:exp-rational:gamma}
\gamma_k = \frac{1}{2^k\cdot k!} \int_{-1}^1 f^{(1)}_d(t)
\od[k]{\ }{t} [(t^2-1)^k] \dif {t} = \frac{(-1)^k}{2^k\cdot k!}  \int_{-1}^1
(t^2-1)^k f^{(k+1)}_d(t) \dif{t} .
\end{equation}
If we let $v \defeq \frac{2d}{(1-t)},$
we obtain, $f_d(t) = e^{d-v}$ and $f^{(1)}_d(t) =
\frac{-1}{(1-t)}ve^{d-v}.$ A simple induction argument generalizes
this to give
\[(1-t)^{k+1} \od[k+1]{\ }{t} f_d(t) = (1-t)^{k+1} f^{(k+1)}_d(t) =
{-e^d} \od[k]{\ }{v} [v^{k+1}e^{-v}].\]
We now invoke the generalized Laguerre polynomials of degree $k$
orthogonal with respect  to the weight function $ve^{-v},$ denoted by $G_k(v)$
and defined to be $\frac{1}{k!}\frac{1}{ve^{-v}} \od[k]{\ }{v} [
v^{k+1}e^{-v}]$ (see~\cite{Szego39}). Hence, simplifying \eqref{eq:exp-rational:gamma}, we obtain
\[\gamma_k = \frac{-e^d }{2^k}
\int_{-1}^1 (t+1)^k \frac{ve^{-v}}{(1-t)} G_k(v) \dif{t} = -e^d
\int_{d}^\infty \left(1-\frac{d}{v}\right)^k e^{-v} G_k(v) \dif{v}.\]
Squaring the above equality, and applying Cauchy-Schwartz, we obtain
\[\gamma_k^2  \le e^{2d} \int_{d}^\infty ve^{-v} (G_k(v))^2 \dif{v} \cdot \int_{d}^\infty
\frac{1}{v}\left(1-\frac{d}{v}\right)^{2k} e^{-v} \dif{v}.\] Now, we
use $\int_{0}^\infty ve^{-v} (G_k(v))^2 \dif{v} = k+1$
(see~\cite{Szego39}), and substitute $v = d(1+z)$ to obtain
\[\gamma_k^2  \le e^d (k+1) \int_{0}^\infty
\frac{z^{2k}}{(z+1)^{2k+1}} e^{-dz} \dif{z}.\]
Plugging this back in \eqref{eq:ssv:error}, we obtain
\begin{align*}
 \left(\inf_{q_d} \sup_{y \in (-1,1]}\left|f_d(y)-q_d(y)\right|
 \right)^2 & \le e^d \int_{0}^\infty
\sum_{k \ge d}  (k+1)(2k+1) \frac{z^{2k}}{(z+1)^{2k+1}} e^{-dz}
\dif{z} .
\end{align*}
We can sum up the series in the above equation for any $z \ge 0$ to
obtain $\sum_{k \ge d} (k+1)(2k+1) \frac{z^{2k}}{(z+1)^{2k+1}}
\lesssim \left(\frac{z}{z+1} \right)^{2d-2} (d^2 + dz + z^2).$ Here $\lesssim$ means that the inequality holds up to an absolute constant. This implies that 
\begin{align*}
  \left(\inf_{q_d} \sup_{y \in (-1,1]}\left|f_d(y)-q_d(y)\right|
  \right)^2 & \lesssim \int_{0}^\infty \left(\frac{z}{z+1}
  \right)^{2d-2} (d^2 + dz + z^2) e^{d-dz} \dif{z} .
\end{align*}
It is a simple exercise to show that for all $z \ge 0,$ the
expression $e^{d-dz+z} \left( \frac{z}{z+1}\right)^{2d-2}$ is
maximized for $z=1$ and, hence, this expression is bounded by $4e\cdot
4^{-d}.$ Thus,
\begin{align*}
  \left(\inf_{q_d} \sup_{y \in (-1,1]}\left|f_d(y)-q_d(y)\right|
  \right)^2 & \lesssim 4^{-d} \int_{0}^\infty (d^2 + dz + z^2) e^{-z}
  \dif{z} \lesssim d^2\cdot 4^{-d},
\end{align*}
which concludes the proof.
\end{proof}

\noindent
In Section \ref{sec:applications:matrix-exp} we also bound the magnitudes of the coefficients in this polynomial and show how to compute arbitrarily good approximations to them efficiently.

\section{Approximating $x^{-1}$ Using Exponentials}
\label{sec:exp-reduction}

In this section we give an
approximation to $x^{-1}$ using a small number of
exponentials. As we see in
Section~\ref{sec:applications:matrix-inverse-reduction}, this
immediately implies a reduction from approximate matrix inversion
(equivalent to approximately solving a linear system) to approximating
the matrix exponential, thus proving that the problems are essentially
equivalent.
\begin{theorem}
\label{thm:inverse-approx}
Given $\eps,\delta \in (0,1],$ there exist $\poly( \log
\nfrac{1}{\eps\delta})$ numbers $0< w_j$ and $t_j =
O(\poly(\nfrac{1}{\eps\delta})),$ such that for all $x \in
[\eps,1],$ we have $(1-\delta)x^{-1} \le \sum_j w_j e^{-t_j x} \le (1+\delta)
x^{-1}.$
\end{theorem}
\noindent
Similar results have appeared in the literature~\cite{BM05,BM10}. The
 proof we present is from~\cite{SachdevaV13}.
The starting point of the proof of Theorem \ref{thm:inverse-approx} is
the identity $x^{-1} = \int_0^{\infty} e^{-xt} \dif{t}.$
The crux of the proof is to discretize this integral to a
\emph{sparse} sum of exponentials. One approach to discretize an
integral to a sum is via the \emph{trapezoidal rule} - by
approximating the area under the integral using \emph{trapezoids} of
small width, say $h$,
\[\int_a^b g(t)\dif{t} \ \approx\  T_{g}^{[a,b],h} \ \defeq \ \frac{h}{2}
\cdot \sum_{j=0}^{K-1} \left(g(a+jh)+g(a+(j+1)h) \right),\] 
where $K \defeq \nfrac{(b-a)}{h}$ is an integer. Applying this rule to
the above integral after truncating it to a large enough interval
$[0,b],$ we obtain the approximation $x^{-1} \approx
\frac{h}{2}\sum_{j=0}^{\nfrac{b}{h}-1} \left(e^{-xjh} + e^{-x(j+1)h}
\right).$ The choice of $h$ determines the discretization of the
interval and, hence, the sparsity of the approximating sum $K$.
Recall that the error must be of the form
\[\textstyle \forall x \in [\eps,1] \qquad \abs{x^{-1} -
\frac{h}{2}\sum_{j} \left(e^{-xjh} + e^{-x(j+1)h} \right)} \le \delta
x^{-1}.\]
For $x = 1,$ we obtain $h \le O_\delta(1)$ and, hence, $K \geq
\Omega_\delta({b}).$ Moreover, the error in truncating the integral is $\int_b^{\infty} e^{-xt} \dif{t} = x^{-1} e^{-bx},$
forcing $b \ge \nfrac{1}{\eps}\cdot \log \nfrac{1}{\delta}$ to be
at most $\nfrac{\delta}{x}$ for all $x \in [\eps,1].$ Thus, this
approach to discretization can only give us a sum which uses
poly$\left(\nfrac{1}{\eps}\right)$ exponentials.

\vspace{2mm}
This suggests that we should select a discretization where $t$
increases much more rapidly with $h$, \emph{e.g.}, exponentially instead of
linearly. This can be achieved by substituting $t=e^y$ in the above
integral to obtain the identity $x^{-1} = \int_{-\infty}^\infty
e^{-xe^y + y} \dif{y}.$ Let $f_x(y) \defeq e^{-xe^y +y}.$ Observe
that $f_x(y) = x^{-1}\cdot f_1(y+\ln x).$ Since we allow the error to
scale with $x^{-1}$ as $x$ varies over $[\eps,1],$ $y$ needs to change
only by an additive $\log \nfrac{1}{\eps}$ to compensate for $x.$ This
suggests that only roughly $\nfrac{1}{h}\cdot \log \nfrac{1}{\eps}$
additional terms are needed above those required for $x=1$ in order for the
approximation to hold for all $x \in [\eps,1],$ giving a logarithmic
dependence on $\nfrac{1}{\eps}.$ We show that discretizing this
integral using the trapezoidal rule, and bounding the error using the
Euler-Maclaurin formula, does indeed give us the above result.

\subsection{Bernoulli Numbers and the Euler-Maclaurin Formula} The Bernoulli
numbers, denoted by $b_i$ for any integer $i \ge 0,$ are a sequence of
rational numbers which, while discovered in an attempt to compute sums
of the form $\sum_{i \geq 0}^k i^j,$ have deep connections to several
areas of mathematics.\footnote{The story goes that when Charles
  Babbage designed the Analytical Engine in the 19th century, one of
  the most important tasks he hoped the Engine would perform was the
  calculation of Bernoulli numbers.} They can be defined recursively: $b_0 = 1,$ and for all $k \ge 1,$ $
\sum_{j=0}^{k-1} {k \choose j} b_j= 0.  $ Given the Bernoulli numbers,
the Bernoulli polynomials are defined to be $B_k(y) \defeq \sum_{j=0}^k
\binom{k}{j} b_j y^{k-j}.$ Using properties of the Bernoulli
polynomials, and a well-known connection to the
Riemann zeta function, we obtain the following bounds
(see~\cite{GrahamKP94}).
\begin{lemma}
\label{lem:bernoulli}
For any non-negative integer $k,$ and for all $y \in [0,1],$
$\frac{|B_{2k}(y)|}{(2k)!} \le \frac{\abs{b_{2k}}}{(2k)!} \le
\frac{4}{(2\pi)^{2k}}.$
\end{lemma}
\noindent
One of the most significant connections in analysis involving the
Bernoulli numbers is the \emph{Euler-Maclaurin formula} which
\emph{exactly} describes the error in approximating an integral by the
trapezoidal rule. For a function $g(y),$ let its $k^\textrm{th}$
derivative be denoted by $g^{(k)}(y) \defeq \od[k]{\ }{y}g(y).$ 
\begin{lemma}
  Given a function $g : \rea \to \rea,$ for any $a < b,$ any $h > 0,$
  and $N \in \nat,$ we have,
\begin{align}
\label{eq:EM}
  \int_a^b g(y)\dif{y} - T_{g}^{[a,b],h} = h^{2N+1} \int_0^K
  \frac{B_{2N}(y-[y])}{(2N)!} g^{(2N)}(a+yh) \dif{y} - \sum_{j=1}^N
  \frac{b_{2j}}{(2j)!}h^{2j} \left(g^{(2j-1)}(b)-g^{(2j-1)}(a)\right),
\end{align}
where $K \defeq \frac{b-a}{h}$ is an integer, and $[\cdot]$ denotes
the integer part. 
\end{lemma}
\noindent
Note that it is really a family of formulae, one for each choice of $N,$ called the \emph{order} of the
formula. The choice of $N$ is influenced by how well behaved the
higher order derivatives of the function are. For example, if $g(y)$
is a polynomial, when $2N > \text{degree}(g),$ we obtain an exact
expression for $\int_a^b g(y)\dif{y}$ in terms of the values of the
derivatives of $g$ at $a$ and $b.$ Since the sparsity of the
approximation is $\Omega(\nfrac{1}{h}),$ for the sparsity to depend
logarithmically on the error parameter $\delta,$ we need to choose $N$ to
be roughly $\Omega(\log \nfrac{1}{\delta})$ so that the first error
term in \eqref{eq:EM} is comparable to $\delta.$

\subsection*{The Proof}
\begin{proof}[Proof of Theorem~\ref{thm:inverse-approx}]
We fix the step size $h,$ approximate the
integral $\int_{-bh}^{bh} f_x(y)\dif{y}$ using the trapezoidal rule
($b$ is a positive integer), and bound the approximation error using
the Euler-Maclaurin formula. We let $b$ go to $\infty$, which allows us
to approximate the integral over $[-\infty,\infty]$ by an infinite sum
of exponentials. Finally, we  truncate this sum to obtain our
approximation. Applying the order $N$ Euler-Maclaurin formula to the
integral $\int_{-bh}^{bh} f_x(y) \dif{y},$ and using
Lemma~\ref{lem:bernoulli}, we obtain,
\begin{align}
\label{eq:EM-error}
\abs{\int_{-bh}^{bh} f_x(y)\dif{y} - T_{f_x}^{[-bh,bh],h}} \le & 4\left(
  \frac{h}{2\pi}
\right)^{2N} \int_{-bh}^{bh} \abs{f^{(2N)}_x(y)} \dif{y} \\
& \qquad + \sum_{j=1}^N 4 \left( \frac{h}{2\pi}\right)^{2j}
\left(\abs{f^{(2j-1)}_x(-bh)}
  +\abs{f^{(2j-1)}_x(bh)}\right). \nonumber
\end{align}
\noindent
Now, the derivatives of the function $f_x(y)$ are well-behaved and
easy to compute. By direct computation, for any $k,$ its
$k^{\textrm{th}}$ derivative $f_x^{(k)}(y)$ is of the form $f_x(y)
p_k(-xe^y),$ where $p_k$ is a degree-$k$ polynomial. Since the
exponential function grows faster than any polynomial, this implies
that for any fixed $k$ and $x,$ $f_x^{(k)}(y)$ vanishes as $s$ goes to
$\pm \infty.$ We let $b$ go to $\infty$ and observe that the
discretized sum converges to $h \sum_{j \in \bz} f_x(jh),$ hence,
\eqref{eq:EM-error} implies that
\begin{align}
\label{eq:inf-sum-error}
\abs{\int_{-\infty}^\infty f_x(y)\dif{y} - h\sum_{j \in \bz} f_x(jh)} \le
4\left( \frac{h}{2\pi} \right)^{2N} \int_{-\infty}^{\infty}
\abs{f^{(2N)}_x(y)} \dif{y}.
\end{align}
Thus, all we need to show is that the function $f_x$ is \emph{smooth
  enough}. There is an easy recurrence between the coefficients of
$p_k$ for various $k$, and it allows us to crudely bound the sum of
their absolute values by $(k+1)^{k+1}$ (Fact 1.3
in~\cite{SachdevaV13}). This, in turn, implies the bound $
\int_{-\infty}^{\infty} \abs{f^{(2N)}_x(y)} \dif{y} \le x^{-1} \cdot
\Theta\left(N\right)^{4N}$ (Lemma 1.4 in~\cite{SachdevaV13}). Thus, we
can choose $h = \Theta(\nfrac{1}{N^2})$ and $N = \Theta( \log
\nfrac{1}{\delta})$ to obtain the following approximation for all $x
> 0$:
\begin{align}
\label{eq:inf-sum-error-final}
\abs{x^{-1} - h\sum_{j \in \bz} e^{jh}\cdot e^{-xe^{jh}}} =
\abs{\int_{-\infty}^\infty f_x(y)\dif{y} - h\sum_{j \in \bz} f_x(jh)}
= O\left(\delta \cdot x^{-1} \right).
\end{align}

\noindent
The final step is to truncate the above discretization. Since the
function $f_x(y)$ is non-decreasing for $y < \log \nfrac{1}{x},$ we
can majorize the lower tail with an integral to obtain $h\sum_{j < A}
f_x(jh) \le \int_{-\infty}^{Ah} f_x(t) \dif{t} = x^{-1} \left(1 -
  e^{-xe^{Ah}}\right).$ Thus, for $A = \floor{-\nfrac{1}{h}\cdot \log
  \nfrac{1}{\delta}},$ we obtain that the lower tail is $O(\delta
  \cdot x^{-1}).$ Similarly, for the upper tail, using that
$f_x(y)$ is non-increasing for $y \ge \log \frac{1}{x},$ for $B =
\ceil{\nfrac{1}{h} \cdot \log \left( \nfrac{1}{\eps} \log
    \nfrac{1}{\delta} \right)},$ we obtain that the upper tail $h\sum_{j
  > B} f_x(jh)$ is $O(\delta \cdot x^{-1}).$ Combining these
tail bounds with~\eqref{eq:inf-sum-error-final}, we obtain
\begin{align*}
  \abs{x^{-1}-h\sum_{j \ge A}^B e^{jh}\cdot e^{-xe^{jh}} }
 = O\left(\delta \cdot x^{-1}\right),
\end{align*}
which completes the proof.
\end{proof}

\section{Applications}\label{sec:applications}
In this section, we present several algorithmic applications of the approximation theory results obtained in the previous sections. All these results are obtained by lifting the approximation results for scalar functions such as $x^s,$ $e^{-x}$ or $x^{-1}$  to the matrix world. Since matrices capture graphs, we often obtain fast algorithms for important graph problems. We start with some basics on matrices and graphs. 
     
\subsection{Matrices and Graphs}
We are primarily concerned with $n \times n$ symmetric matrices over the reals. 
A fundamental theorem in linear algebra (see \cite[Chapter 1]{Vishnoi12}) 
asserts  that every symmetric matrix $A \in \mathbb{R}^{n \times n}$ has $n$ real eigenvalues along with eigenvectors that can be chosen to be orthogonal. Thus, $A$ can be written as $U\Lambda^\top U$ where the columns of $U$ are the eigenvectors of $A$ and $\Lambda$ is the diagonal matrix corresponding to its eigenvectors. $A$ is said to be positive semidefinite (PSD) if all its eigenvalues are non-negative and positive definite (PD) if all its eigenvalues are strictly positive. The spectral norm of a matrix $A$ is its $2 \rightarrow 2$ norm, 
which is defined to be $\sup_{x \neq \vec{0}} \frac{\|Ax\|_2}{\|x\|_2}.$ Thus, all eigenvalues of $A$ are bounded in absolute value by the spectral norm of $A$.  For a  PSD matrix, its norm is  its largest eigenvalue. Henceforth, $\| \cdot \|$ is used to denote the $\ell_2$ norm for vectors and the spectral norm for matrices. 

\vspace{2mm}
For a function $f:\mathbb{R} \mapsto \mathbb{R}$ and a real symmetric matrix $A,$ one can define $f(A)$ as follows. First, for a diagonal matrix $\Lambda$ let $f(\Lambda)$ denote the diagonal matrix where the $(i,i)$th entry is $f(\Lambda_{i,i}).$ Then, $f(A)$ is defined to be $Uf(\Lambda)U^\top$ where $U\Lambda U^\top$ is the spectral decomposition of $A.$ 
Thus, a polynomial $p(x) = \sum_{i=0}^d c_i x^i,$  when applied to $A,$ is a matrix $p(A)$ which can be seen to be $\sum_{i=0}^d c_iA^i$ since $U^\top U = I.$   
Moreover, $\exp(A)$ or $e^A$ is  $\sum_{k=0}^\infty \frac{A^k}{k!}.$

\vspace{2mm}
For an $n \times n$ matrix $A$ and a vector $v,$ often we are  interested in the solution to the system of equations $Ax=v.$ We  only consider the case when either $A$ is invertible or $v$ lies in the span of the columns of $A.$ In either case, with a slight abuse of notation, we denote the solution by $x=A^{-1} v.$ 

\vspace{2mm}
Finally, we are in interested in undirected graphs $G=(V,E)$ with $n \defeq |V|$ vertices and $m \defeq |E|$ edges. The edges of the graph may have positive weights and this is captured by the {\em adjacency} matrix $A$ of the graph; an $n \times n$ matrix where $A_{i,j}$ is the  weight of the edge between $i$ and $j.$ We  assume that the graph has no self-loops and, hence, $A_{i,i}=0$ for all $i.$ Since the graph is undirected, $A$ is symmetric and has $m$ non-zero entries. Let $e_i$ denote the vector with $1$ in the $i$th coordinate and $0$ elsewhere. The matrix 
$L \defeq \sum_{e=i,j} A_{i,j} (e_i-e_j)(e_i-e_j)^\top$ is called the {\em combinatorial Laplacian} of $G.$ If $D$ is the diagonal matrix with $D_{i,i} \defeq \sum_{j \neq i} A_{i,j},$ then $L=D-A.$ The Laplacian of a graph $L$ is always PSD; $ L \succeq 0.$ 

\subsection{Simulating Random Walks and Finding Sparse Cuts}
\label{sec:applications:sim-rw}
Consider a graph $G=(V,E)$ with $|V|=n, |E|=m,$ and let $A$ denote its
adjacency matrix. The simple random walk on such a graph corresponds
to the process where, starting at a vertex $i,$ one selects a vertex
$j$ with probability proportional to $A_{i,j}$, and then repeats with
$j$ as the starting vertex. Suppose we select an initial vertex from
the probability distribution $v \in \rea^n$ and perform an $s$-step
random walk, the probability distribution of the vertex after $s$
steps of this random walk is given by $\tilde{W}^s v$ \footnote{The
  convention in the Markov Chains literature is to express the
  probability distribution as a row vector $v^\top$ instead, giving
  the probability after $s$ steps as $v^\top W^s.$ We will use the
  column vector convention. The only resulting change is that the walk
  matrix is replaced by its transpose everywhere.}, where $\tilde{W}
\defeq D^{-1} A.$ Computing such distributions, sometimes starting
from arbitrary real vectors rather than probability vectors, is a
fundamental problem that finds many applications, for instance in
finding sparse cuts in graphs as explained below; often, a good enough
approximation to such a distribution suffices.

\subsubsection{Quadratically Faster Random Walks}\label{sec:rw}
A simple way to compute $\tilde{W}^sv$   is to multiply the matrix
$\tilde{W}$ with $v$ a total of $s$ times, which requires $O(ms)$ operations. We now show that, as
an immediate application of the polynomial approximations to $x^s$ that
we developed in Section~\ref{sec:xk-approx}, we can approximate this
distribution using roughly $\sqrt{s}$ multiplications with $\tilde{W}.$ First,
we extend Theorem~\ref{thm:xk-approx} to matrices.
\begin{theorem}[Corollary to Theorem~\ref{thm:xk-approx}]
\label{cor:matrix-approx-poly}
For a symmetric $M$ with $\norm{M} \leq 1$, a positive integer $s$ and
any $\delta > 0,$ define $d \defeq \ceil{\sqrt{2s \log
    \nfrac{2}{\delta}}}$. Then, the degree-$d$ polynomial $p_{s,d}(M),$
defined by \eqref{eq:approx-poly} satisfies
$\norm{M^s - p_{s,d}(M)} \leq \delta.$
\end{theorem}
\begin{proof}
Let $\{\lambda_i\}_i$ be the eigenvalues of $M$ with $\{u_i\}_i$ as the
set of corresponding orthogonal eigenvectors. Since $M$ is symmetric and $\norm{M}\le 1,$
we have $\lambda_i \in [-1,1]$ for all $i.$ Thus,
Theorem~\ref{thm:xk-approx} implies that for all $i,$ $|\lambda_i^s-
p_{s,d}(\lambda_i)| \le \delta.$ Note that if $\lambda_i$ is an
eigenvalue of $M,$ then $\lambda_i^{s}$ is the corresponding
eigenvalue of $M^{s}$ and $p_{s,d}(\lambda_i)$ is that of $p_{s,d}(M)$
with the same eigenvector. Hence, we have $\norm{M^s - p_{s,d}(M)} =
\norm{\sum_i (\lambda_i^s - p_{s,d}(\lambda_i))u_iu_i^\top} = \max_i
|\lambda_i^s - p_{s,d}(\lambda_i)| \le \delta.$
\end{proof}

\noindent
When we try to apply this theorem to $\tilde{W}$ we face the obvious problem that $\tilde{W}$ is not necessarily symmetric.  This can be handled  by considering the matrix $W \defeq D^{-\nfrac{1}{2}} \tilde{W} D^{\nfrac{1}{2}}$, which is symmetric. Thus, $\tilde{W}^s v = D^{\nfrac{1}{2}} W^s D^{-\nfrac{1}{2}}v.$ For now, we consider the case that $G$ is $d$-regular, \emph{i.e.}, $D=d \cdot I$ for some $d.$  In this case $\tilde{W}=W.$  Further, it can be seen that $\|W\| \leq 1$ since $W$ is a doubly stochastic matrix. 

\vspace{2mm}
Note that if we can compute the coefficients of $p_{s,d}$ efficiently, then we can quickly compute $p_{s,d}(W)v$ for $d = \ceil{\sqrt{2s \log  \nfrac{2}{\delta}}}$. Thus, appealing to the theorem above, we obtain an efficient $\delta$ approximation to $W^sv,$ \emph{i.e.}, $\norm{W^sv - p_{s,d}(W)v} \le \delta \norm{v} \leq \delta.$ 
In order to compute the coefficients, first  observe that we do not need to
explicitly compute the coefficients of the polynomial $p_{s,d}$ since we
can use the expansion of $p_{s,d}$ in terms of Chebyshev polynomials as in \eqref{eq:approx-poly} and the recursive definition of Chebyshev polynomials from
\eqref{eq:prelims:chebyshev-recurrence} to compute the
vectors $T_0(W)v, \ldots, T_d(W)v$ using only $d$ multiplications with
the matrix $W.$ 

\vspace{2mm}
The expansion of $p_{s,d}$ in terms of Chebyshev polynomials given by \eqref{eq:approx-poly}  implies that the non-zero
coefficients are binomial coefficients up to a scaling
factor. For instance,  assuming that $s$ is even, the coefficient of
$T_1(\cdot)$ is $2^{-s+1}\binom{\nfrac{s}{2}}{\nfrac{s}{2}+j}.$
\emph{Prima facie}, computing these binomial coefficients requires $O(s)$
multiplications and divisions which is worse than the trivial $O(ms)$ time algorithm to compute $W^sv.$  However, since the non-zero
coefficients are scaled binomial coefficients, if $c_i$ is the
coefficient of $T_i,$ the ratios $\nfrac{c_i}{c_0}$ are rational
numbers that we can compute explicitly. We also note that the sum of
the coefficients of $T_0(\cdot),\ldots, T_d(\cdot)$ in $p_{s,d}$ lies
between 1 and $1-\delta.$ Thus, we can explicitly compute $\sigma \defeq
\sum_{i=0}^d c_i/c_0,$ and $\nfrac{1}{\sigma}\cdot \nfrac{c_i}{c_0} =
\nfrac{c_i}{\alpha}, $ where $\alpha \defeq
\Pr_{Y_1,\ldots,Y_s}\left[\left|D_s\right| \le d\right] = \sum_{i=0}^d
c_i$ is the sum of coefficients of $p_{s,d}.$ Hence, we know the
coefficients in the Chebyshev expansion of the polynomial $\alpha^{-1}
\cdot p_{s,d}(\cdot),$ and it satisfies $\sup_{x \in [-1,1]}
\left|\alpha^{-1} \cdot p_{s,d}(x) - x^s\right| \le
\nfrac{\delta}{(1-\delta)} = O(\delta).$\footnote{  An important issue we need to note is the bit length of the numbers
  involved. Even though it is not possible to store these numbers
  precisely, here we  show that few bits to store each of these
  numbers are sufficient. Assume that we store each of the numbers in
  $b$-bit registers. All the numbers involved in computing the ratios
  of successive coefficients are $O(s),$ thus we need $b = \Omega(\log
  s).$ Each of these ratios can be computed to an accuracy of
  $O(2^{-b}),$ and since there are $O(d)$ multiplications/divisions
  involved, we can compute all of $\nfrac{c_i}{c_0}$ up to an accuracy
  of $O(d2^{-b}).$ Hence, the absolute error in $\sigma$ is at most
  $O(d^2 2^{-b}).$
 This implies that if $d^22^{-b} =
  O(\delta),$ the error in the estimate $u$ is at most $O(\delta)
  \norm{v}.$ Thus, $b = \Theta(\log \nfrac{s}{\delta})$ suffices.}
 We summarize this in the
following theorem.
\begin{theorem}\label{thm:simrw}
Let ${W}$ be the random walk matrix for a  regular graph $G$ with $n$
vertices and $m$ edges . Then, for any positive integer $s,$ starting
distribution $v,$ and $\delta \in (0,\nfrac{1}{2}]$, there is an
algorithm that computes a vector ${w}$ such that $\norm{W^sv-{w}} \le
\delta  $ in $O\left(\left(m+n\right) \sqrt{s \log
\nfrac{1}{\delta}}\right)$ arithmetic operations.
\end{theorem}

\noindent
Theorem \ref{thm:simrw} can be easily generalized to a reversible
irreducible Markov chain with transition matrix $P$ and stationary
distribution $\pi.$ Let $\Pi$ be the diagonal matrix defined by
$\Pi(i,i) = \pi(i),$ the matrix $\Pi^{\nfrac{1}{2}} P^\top
\Pi^{-\nfrac{1}{2}}$ is symmetric and has norm at most 1 and, hence, we
can apply the above algorithm with $W = \Pi^{\nfrac{1}{2}} P^\top
\Pi^{-\nfrac{1}{2}}$ and the vector $\Pi^{\nfrac{1}{2}} v,$ and obtain
a vector $u$ with the approximation guarantee $\norm{(P^\top)^s v-{u}}
\le \delta \sqrt{\frac{\max_i \pi(i)}{\min_i \pi(i)}} \norm{v}$ in
$O\left(\left(t_P+n\right) \sqrt{s \log \nfrac{1}{\delta}}\right)$
arithmetic operations, where $t_P$ is the cost of multiplying the
matrix $P^\top$ with a given vector.

\subsubsection{Finding Sparse Cuts}\label{sec:sparsestcut}
We now outline how we can use the algorithm in the proof of Theorem \ref{thm:simrw} to speed up an algorithm to find {\em sparse}  cuts in a graph. For a graph $G = (V,E)$ with adjacency
matrix $A,$ a cut $S \subseteq V$ is said to have {\em sparsity} or {\em conductance}
\[\phi(S) \defeq \frac{\sum_{i \in S} \sum_{j \not\in S}
  A(i,j)}{\min\left(\sum_{i \in S}\sum_{j \in V} A(i,j), \sum_{i
      \not\in S}\sum_{j \in V} A(i,j)\right)}.\] The conductance of a
graph, $\phi \defeq \min_{S \subseteq V} \phi(S)$, gives a measure of
how interconnected a graph and is an important problem theory and practice,  see \cite[Chapter 5]{Vishnoi12} 
for a detailed discussion on this problem. It is also NP-hard to find the cut of least conductance and, hence, one has to be satisfied with algorithms that compute cuts whose sparsity is close to that of the sparsest cut. A celebrated result of  Cheeger \cite{Cheeger70} 
and Alon and Milman \cite{AlonM85}
relates the second smallest eigenvalue of the Laplacian $L$ of $G,$ denoted $\lambda_2(L)$  to the conductance of the graph. Often referred to as Cheeger's inequality, the result, stated here for $d$-regular graphs,  asserts that 
$  \phi \leq  O \left( \sqrt{\nfrac{ \lambda_2}{d}}\right).$ Let $\lambda \defeq  \nfrac{\lambda_2}{d}$ be the normalized {\em spectral gap} and $\mathcal{L} \defeq \frac{1}{d} L$ be the normalized Laplacian.
In fact, a cut of conductance $O \left( \sqrt{\lambda}\right)$ can be recovered from the second eigenvector of $L$ and, thus, algorithmically, it is sufficient to compute the second eigenvector of $L.$  Mihail \cite{Mihail89} 
proved a stronger version of this theorem which showed how to produce a cut of sparsity at most $O(\sqrt{\lambda'})$ from any  vector $v$  (orthogonal to the all ones vector) such that $\frac{v^\top \mathcal{L} v}{v^\top v} = \lambda'.$ 
Note that for $d$-regular graphs, the all ones vector is an eigenvector of $L$ with eigenvalue $0.$ Hence, the second eigenvector is orthogonal to this vector.

\vspace{2mm}
We show how, as a direct consequence to  Theorem \ref{thm:simrw}, we can produce a vector $u$ such that $\frac{u^\top  \mathcal{L} u}{u^\top u} \le O(\lambda)$ giving us an algorithm to find a cut of sparsity at most $O(\sqrt{\lambda})$ in time roughly ${O}(\nfrac{m}{\sqrt{\lambda}}).$ This gives a   quadratically better dependence in $\lambda$ than the standard algorithm using the Power method.  Formally, we prove the following theorem.

\begin{theorem}
  Given an undirected regular graph $G$ with normalized spectral gap
  $\lambda,$ we can find a cut of conductance $O(\sqrt{\lambda})$ with
  probability at least $\nfrac{1}{3}$ using
  $O\left(\nfrac{m}{\sqrt{\lambda}} \cdot \log
    \nfrac{n}{\lambda}\right)$ operations.
\end{theorem}
\begin{proof}
  We use the  algorithm in the proof of Theorem \ref{thm:simrw} to approximate $W^sv$ where $W=\frac{1}{d} A$ is the
  random walk matrix, $v$  a random unit vector orthogonal to the all
  ones vector, with parameters $s$ and $\delta$ (the required $s$ and
  $\delta$ is determined later). 
Note that $\lambda$ is the second smallest eigenvalue of $I-W$ and, hence, $1-\lambda$ is the second largest eigenvalue of $W.$ Let $v_2$ be the corresponding unit eigenvector.

\vspace{2mm}
Let $u$ be the approximating
  vector as obtained from Theorem \ref{thm:simrw}. Let  $u^\star \defeq W^sv,$ and let
  $\Delta \defeq u - u^\star.$ Thus, we know that $\norm{\Delta} \le
  \delta.$ Since $v$ was
  chosen as a random unit vector orthogonal to the uniform
  distribution, with probability at least $\nfrac{2}{3}$ we have,
  $v_2^\top v \ge \frac{1}{3\sqrt{n}}.$ This implies,
\[u^{\star \top} u^\star \ge \frac{1}{2}u^{\star \top} \mathcal{L} u^\star \ge
\lambda (1-\lambda)^{2s} \frac{1}{18n}.\]
We choose $\delta \defeq \sqrt{\lambda (1-\lambda)^{2s} \frac{2}{9n}},$
implying $u^{\star \top}  \mathcal{L} u^\star \ge \frac{1}{4}\delta^2.$
Thus, we have,
\[\frac{u^\top  \mathcal{L} u}{u^\top u} = \frac{(u^\star + \Delta)^\top  \mathcal{L} (u^\star + \Delta)}{(u^\star + \Delta)^\top (u^\star + \Delta)} \le \frac{2
  (u{^{\star\top}} \mathcal{L} u^\star + \Delta^\top \mathcal{L}
  \Delta)}{\left(\norm{u^\star} - \norm{\Delta}\right)^2} \le
\frac{2(u^{\star \top} \mathcal{L} u^\star +
  2\delta^2)}{\left(\norm{u^\star} - \delta\right)^2} \le \frac{12
  u^{\star\top} \mathcal{L} u^\star}{u^{\star \top} u^\star},\] where
the first inequality uses $(u^\star + \Delta)^\top \mathcal{L}
(u^\star + \Delta) \le 2 (u{^{\star\top}} \mathcal{L} u^\star +
\Delta^\top \mathcal{L} \Delta),$ which is the same as $0 \le (u^\star
- \Delta)^\top \mathcal{L} (u^\star - \Delta)$ after rearranging. As
dictated by the Power method to approximate the spectral gap, we
choose $s \defeq \ceil{\frac{\log ({9n}/{\lambda})}{2\log
    (1/(1-\lambda))}}$ to ensure $\frac{ u^{\star\top} A
  u^\star}{u^{\star \top} u^\star} \ge 1-O(\lambda)$
(see~\cite[Chapter 8]{Vishnoi12} for a proof). This
implies, \[\frac{u^{\star\top} \mathcal{L} u^\star}{u^{\star \top}
  u^\star} = 1-\frac{u^{\star\top} A u^\star}{u^{\star \top} u^\star}
= O(\lambda).\] Thus, $\frac{u^\top \mathcal{L} u}{u^\top u} \le
O(\lambda).$ We note that since $v$ is orthogonal to the all ones
vector, the vector $u$ returned is also orthogonal to the all ones
vector. Hence, by Mihail's theorem, we can round $u$ to find a cut of
conductance $O(\sqrt{\lambda}).$

The running time for this procedure is dominated by the time required
to compute $u,$ which requires $O\left((t_A + n)\sqrt{s \log
    \nfrac{1}{\delta}}\right)$ operations. We note that $\delta =
\Omega\left(\nfrac{\lambda}{n}\right),$ which implies that the total number
of operations required is $O\left(\nfrac{m}{\sqrt{\lambda}} \cdot \log
  \nfrac{n}{\lambda}\right).$
\end{proof}

\subsection{Solving Linear Equations}
\label{sec:applications:cg}
Given a matrix $A \in \rea^{n
  \times n}$ and a vector $v \in \rea^n,$  our goal is to find a
vector $x\in \rea^n$ such that $Ax = v.$ 
The exact solution $x^\star \defeq A^{-1}v$ can be computed by Gaussian
elimination, but the fastest known implementation requires
$O(n^{2.737})$ time. For many applications, the number of non-zero entries in $A$ (denoted by  $m$), or its {\em sparsity}, is much smaller than $n^2$ and, ideally, we would like linear solvers which run in time  $\tilde{O}(m)$~\footnote{The $\tilde{O}$ notation
  hides polynomial factors in $\log n.$}, roughly the time it takes to
multiply a vector with $A.$  While we are far from this goal for general matrices,  \emph{iterative methods}, based on techniques such as  gradient descent or the Conjugate Gradient method reduce the problem of solving a system of linear equations to the computation of a small number of matrix-vector products with the matrix $A$ when $A$ is symmetric and positive definite (PD).   The solutions these methods produce are, in general, approximate which suffice for most applications.
While the running time of the gradient descent-based method varies linearly with the condition number of $A,$ that of  the Conjugate Gradient method   depends on the square-root of the condition number; the quadratic saving occurring  precisely because of the $\sqrt{s}$ degree polynomials approximating $x^s.$

\subsubsection{A Gradient Descent Based Linear Solver} 

The gradient descent method is a general method to solve convex programs; here we only focus on its application to linear systems.  The PD assumption on $A$ allows us to formulate  the problem of solving $Ax=v$  as a convex programming problem:
Let the squared $A$-norm of the
error $x-x^\star$ be $f_A(x) \defeq \norm{x-x^\star}_A^2 \defeq
(x-x^\star)^\top A (x - x^\star) = x^\top A x - 2x^\top v + x^{\star
  \top}Ax^{\star}$, and find the vector $x$ that minimizes $f_A(x)$.
Since $A$ is symmetric and PD,
 this is a convex function, and has a
unique minimizer $x = x^\star.$

\vspace{2mm}
When minimizing $f_A$, each iteration of the gradient descent method is as follows: Start from the current estimate of $x^\star$, say $x_{t},$ and move
along the direction of maximum rate of decrease of the function $f_A,$
\emph{i.e.}, against its gradient, to the point that minimizes the
function along this line. Thus, $x_{t+1} = x_{t} - \alpha_t
\nabla f_A(x_{t}) = x_{t} - \alpha_{t} (Ax_{t}-v).$ If we
define the \emph{residual} $r_{t} \defeq v-Ax_t,$ we can easily
compute the $\alpha_{t}$ that minimizes $f_A$ as $\frac{r_{t}^\top
  r_{t}}{r_{t}^\top A r_{t}}.$ Substituting this value of
$\alpha_{t},$ and using $x^\star-x_{t} = A^{-1}r_{t},$ we obtain
\[\|x_{t+1}-x^\star\|_A^2 = \|x_{t}-x^\star\|_A^2 -
\frac{(r_{t}^\top r_{t})^2}{r_{t}^\top A r_{t}}
=\|x_{t}-x^\star\|_A^2 \left( 1- \frac{r_{t}^\top r_{t}}{r_{t}^\top A
    r_{t}} \cdot \frac{r_{t}^\top r_{t}}{r_{t}^\top A^{-1} r_{t}}
\right). \] For any $z,$ we have $z^\top A z \le \lambda_1 z^\top z$
and $z^\top A^{-1} z \le \lambda_n^{-1} z^\top z,$ where $\lambda_1$
and $\lambda_n$ are the smallest and the largest eigenvalues of $A$
respectively. Thus, $\| x_{t+1}-x^\star\|_A^2 \le
(1-\kappa^{-1})\|x_{t}-x^\star\|_A^2,$ where $\kappa \defeq
\nfrac{\lambda_1}{\lambda_n}$ is the \emph{condition number} of $A.$
Hence, assuming we start with $x_{0} = \mathbf{0},$ we can find an $x_{t}$
such that $\|x_{t} - x^\star\|_A \le \delta \|x^\star\|_A$ in
approximately $\kappa \log \nfrac{1}{\delta}$ iterations, with the cost of
each iteration dominated by $O(1)$ multiplications of the matrix $A$
with a given vector (and $O(1)$ dot product computations). Thus, this gradient descent-based method allows us to compute a $\delta$ approximate solution to $x^\star$ in time $O((t_A+n) \kappa \log \nfrac{1}{\delta}).$

\subsubsection{The Conjugate Gradient Method} 
Suppose we run the gradient descent-based 
method described in the previous section for $k$ iterations. Observe that
at any step $t,$ we have $x_{t+1} \in \Span\{x_{t},Ax_{t},v\}.$
Hence, for $x_{0}={\mathbf{0}}$, it follows by induction that $x_{k} \in
\Span\{v,Av,\ldots,A^kv\}.$ 
The running time of the gradient descent-based 
method is dominated by the time required to compute a basis for
this subspace. However, this vector $x_{k}$ may not be a vector from this subspace that minimizes $f_A.$ 
On the other hand, the essence of the Conjugate Gradient method is that it finds the vector in this subspace that minimizes $f_A$, 
in essentially the same amount of time required by $k$ iterations of the gradient
descent-based method.
We must address two important questions about the Conjugate Gradient method: (1) Can
the best vector be computed efficiently?, and (2) What is the
approximation guarantee achieved after $k$ iterations? We  show that the best vector can be found efficiently, and prove, using the polynomial
approximations to $x^k$ from Section~\ref{sec:xk-approx}, that
the Conjugate Gradient method achieves a quadratic improvement  over the gradient descent-based method in terms of its dependence on the condition number of $A.$

\vspace{2mm} Let us consider the first question. Let
$\{v_0,\ldots,v_k\}$ be a basis for $\calK = \Span\{v, Av, \ldots,
A^kv\}$ (called the \emph{Krylov subspace of order $k$}). Hence, any
vector in the subspace can be written as $\sum_{i=0}^k \alpha_i v_i.$
Our objective then becomes $\|x^\star - \sum_i \alpha_i v_i\|_A^2 =
(\sum_i \alpha_i v_i)^\top A (\sum_i \alpha_i v_i) - 2 (\sum_i
\alpha_i v_i)^\top v + \norm{x^{\star}}_A^2.$ Solving this
optimization problem for $\alpha_i$ requires matrix inversion, the
very problem we set out to mitigate. The crucial observation is that
if the $v_i$s are $A$-orthogonal, \emph{i.e.}, $v_i^\top A v_j = 0$
for $i \neq j,$ then all the cross-terms disappear. Then, $\|x^\star -
\sum_i \alpha_i v_i\|_A^2 = \sum_i (\alpha_i^2 v_i^\top A v_i -2
\alpha_i v_i^\top v) + \norm{x^{\star}}_A^2,$ and we can explicitly
obtain the best solution since $\alpha_i = \frac{v_i^\top v}{ v_i^\top
  A v_i}$ as in the gradient descent-based method.

\vspace{2mm}
Hence, if we can construct an $A$-orthogonal basis
$\{v_0,\ldots,v_k\}$ for $\calK$ efficiently, we do at least as well as the 
gradient descent-based method. If we start with an arbitrary
set of vectors and try to $A$-orthogonalize them via the Gram-Schmidt
process (with inner products with respect to  $A$), we  need to compute $k^2$
inner products and, hence, for large $k,$ it is not more efficient
than the gradient descent-based method. An efficient construction of such a
basis is one of the key ideas here. We  proceed iteratively,
starting with $v_{0} = v.$ At the $i^\textrm{th}$ iteration, we
compute $Av_{i-1}$ and $A$-orthogonalize it with respect to  $v_0,\ldots,v_{i-1},$
to obtain $v_i.$ It is trivial to see  that the vectors $v_0,\ldots,v_k$ are 
 $A$-orthogonal. Moreover, it is not difficult to see that for every
$i,$ we have $\Span\{v_0,\ldots,v_i\} = \Span\{v,Av,\ldots,A^i v\}.$
Now,  since $Av_j \in
\Span\{v_0,\ldots,v_{j+1}\}$ for every $j,$ and $A$ is symmetric,
$A$-orthonormality of the vectors implies $v_i^\top A (Av_j) =
v_j^\top A(Av_i) = 0$ for all $j$ such that $j+1 < i.$ This implies
that we need to $A$-orthogonalize $Av_i$ only to vectors $v_i$ and
$v_{i-1}.$ Hence, the time required for constructing this basis is
dominated by $O(k)$ multiplications of the matrix $A$ with a given
vector, and $O(k)$ dot-product computations.

\vspace{2mm}
Hence we can find the best vector in the Krylov subspace efficiently enough. We now
analyze the approximation guarantee achieved by this vector. Note that
the Krylov subspace $\calK = \Span\{v,Av,\ldots,A^kv\}$ consists of
exactly those vectors which can be expressed as $\sum_{i=0}^k \beta_i
A^i v = p(A)v,$ where $p$ is a degree-$k$ polynomial defined by the
coefficients $\beta_i.$ Let $\Sigma_k$ denote the set of all
degree-$k$ polynomials. Since the output vector $x_{k}$ is the
vector in the subspace that achieves the best possible error
guarantee, we have
\begin{align*}
  \|x_{k}-x^\star\|_A^2 & = \inf_{x \in \calK} \|x^\star-x\|_A^2 =
  \inf_{p \in \Sigma_k} \|x^\star - p(A)v\|_A^2 \le \|x^\star\|_A^2 \cdot
  \inf_{p \in \Sigma_k} \|I - p(A)A\|^2.
\end{align*}
 Observe that the last
expression can be written as $\|x^\star\|_A^2 \cdot \inf_{q \in
  \Sigma_{k+1}, q(0)=1} \|q(A)\|^2,$ where the minimization is now
over degree-$(k+1)$ polynomials $q$ that evaluate to 1 at 0. Since $A$
is symmetric and, hence, diagonalizable, we know that $\|q(A)\|^2 =
\max_{i} |q(\lambda_i)|^2 \le \sup_{\lambda \in [\lambda_n,\lambda_1]}
|q(\lambda)|^2,$ where $0 < \lambda_n \le \cdots \le \lambda_1$
denote the eigenvalues of the matrix $A.$ Hence, in order to prove
that an error guarantee of $\|x_{k}-x^\star\|_A \le \delta
\|x^\star\|_A$ is achieved after $k$ rounds, it suffices to show that
there exists a polynomial of degree $k+1$ that takes value 0 at 1, and
whose magnitude is less than $\delta$ on the interval
$[\lambda_n,\lambda_1].$

\vspace{2mm}
As a first attempt, we consider the degree-$s$ polynomial $q_0(x)
\defeq \left( 1 - \nfrac{2x} {(\lambda_1 + \lambda_n)}\right)^{s}.$
The maximum value attained by $q_0$ over the interval
$[\lambda_n,\lambda_1]$ is $\left( \nfrac{(\kappa - 1)}{(\kappa +1)}
\right)^s.$ Hence, $d_0 \defeq \ceil{\kappa \log \nfrac{1}{\delta}}$
suffices for this value to be less than $\delta.$ Or equivalently,
approximately $\kappa \log \nfrac{1}{\delta}$ rounds suffice for error
guarantee $\|x-x^\star\|_A \le \delta \|x^\star\|_A,$ recovering the
guarantee provided by the gradient descent-based method.

\vspace{2mm}
However, for a better guarantee, we can apply the polynomial approximation to $x^{d_0}$
   developed in Section~\ref{sec:xk-approx}. Let $z \defeq 1 -
\nfrac{2x} {(\lambda_1 + \lambda_n)}.$ Hence, $q_0(x) = z^s.$ As $x$
ranges over $[0,\lambda_n+\lambda_1],$ the variable $z$ varies over
$[-1,1].$ Theorem~\ref{thm:xk-approx} implies that for $d \defeq
\ceil{\sqrt{2d_0 \log \nfrac{2}{\delta}}},$ the polynomial
$p_{d_0,d}(z)$ approximates the polynomial $z^{d_0}$ up to an error of
$\delta$ over $[-1,1].$ Hence, the polynomial $q_1(x) \defeq
p_{d_0,d}\left(z\right)$ approximates $q_0(x)$ up to $\delta$ for all
$x \in [0,\lambda_1+\lambda_n].$ Combining this with the observations from
the previous paragraph, $q_1(x)$ takes value at most $2\delta$ on the
interval $[\lambda_n,\lambda_1],$ and at least $1-\delta$ at 0. Thus,
the polynomial $\nfrac{q_1(x)}{q_1(0)}$ is a polynomial of degree $d =
O(\sqrt{\kappa} \log \nfrac{1}{\delta}) $ that takes value 1 at 0, and
at most $\nfrac{2\delta}{(1-\delta)} = O(\delta)$ on the interval
$[\lambda_n,\lambda_1].$ Or equivalently, $O(\sqrt{\kappa} \log
\nfrac{1}{\delta})$ rounds suffice for an error guarantee
$\|x-x^\star\|_A \le O(\delta) \|x^\star\|_A,$ which gives a quadratic
improvement over the guarantee provided by the gradient descent-based method.
We summarize the guarantees of the Conjugate Gradient method in the
following theorem:
\begin{theorem}
Given an $n \times n$ symmetric matrix $A \succ 0,$ and a vector $v
\in \rea^n,$ the Conjugate Gradient method can find a vector $x$ such
that $\|x-A^{-1}b\|_A \le \delta \|A^{-1}b\|_A$ in time
$O((t_A+n)\cdot \sqrt{\kappa(A)} \log \nfrac{1}{\delta}),$ where $t_A$
is the time required to multiply $A$ with a given vector, and
$\kappa(A)$ is the condition number of $A.$
\end{theorem}

\noindent
We note that this proof of the guarantee of the Conjugate Gradient
method is different from the traditional proof, which directly proves
that the polynomial $T_d(1 - \nfrac{2x} {(\lambda_1 + \lambda_n)})$
for $d = O(\sqrt{\kappa} \log \nfrac{1}{\delta})$ is such that it takes value 0 at 1, and is smaller than $\delta$ in magnitude on
the interval $[\lambda_n,\lambda_1]$ (see,
\emph{e.g.},~\cite{Vishnoi12}).

\subsection{Computing Eigenvalues via the Lanczos Method}\label{sec:lanczos}
The Conjugate Gradient method is one of several methods that work
with the Krylov subspace, collectively called \emph{Krylov subspace
  methods}. Another Krylov subspace method of particular interest is the Lanczos
method, which is typically employed for approximating the eigenvalues and
eigenvectors of a symmetric matrix, see \cite{parlett1980symmetric} for an extensive discussion.
In this section, we present the Lanczos method for approximating the largest eigenvalue of a symmetric matrix and show how existence of good polynomial approximations to $x^s$ allow us to easily improve upon the power method. We conclude this section with a brief discussion on the  generalizations of the Lanczos method  to computing quantities such as $f(A)v.$ Here, rather straightforwardly,  the {\em existence} of low degree polynomial approximations to $f(\cdot)$ in the interval containing the eigenvalues of $A$ imply fast algorithms for computing good approximations to $f(A)v$ quickly. 
For simplicity, in this section we  assume that the matrix is PSD. 

\vspace{2mm}
We start by recalling the variational characterization of eigenvalues: The
largest eigenvalue of $A$ is equal to the maximum value of the
Rayleigh quotient $\frac{w^\top A w}{w^\top w}$ over all non-zero
vectors $w.$ The power method (see~\cite[Chapter 8]{Vishnoi12}) tells us that for
a unit vector $v$ picked uniformly at random, with constant
probability, the vector $A^sv$ achieves a Rayleigh quotient of at
least $(1-\delta)\lambda_1$ for $s $ roughly $\nfrac{1}{\delta},$ where $\lambda_1$ is the largest eigenvalue of $A.$ 
The Lanczos method {\em essentially} finds the vector in the Krylov subspace
$\calK \defeq \{v,Av,\ldots,A^kv\}$ that maximizes the Rayleigh
quotient. We prove below, again using the polynomial approximations to
$x^s$ from Section~\ref{sec:xk-approx}, that in order to find a vector with Rayleigh
quotient at least $(1-\delta)\lambda_1$, it suffices to choose $k$
to be approximately  $\nfrac{1}{\sqrt{\delta}}$. Such a result was proven
in~\cite{KuczynskiW92}. We present a simpler proof here with a
slightly worse bound.

\vspace{2mm}
Let $\lambda_1 \ge \cdots \ge \lambda_n$ be the eigenvalues of $A,$
and let $u_1,\ldots,u_n$ be the corresponding eigenvectors. Let $\delta >
0$ be a specified error parameter. Pick $v$ to be a unit vector chosen
uniformly at random. Assume that $v$ can be expressed in the
eigenbasis for $A$; \emph{i.e.},  $v=\sum_{i=0}^n \alpha_i u_i.$
Let $\{v_0, \ldots,v_k\}$ be any orthonormal basis for the Krylov
subspace $\calK.$\footnote{Later on we show how to construct such a
  basis quickly, similar to the case of the Conjugate Gradient
  method.} Let $V$ denote the $n \times (k+1)$ matrix whose $i$th
column is $v_i.$ Thus, $V^\top V = I_{k+1}$ and $VV^\top$ is the
orthogonal projection on to $\calK.$ Let $T \defeq V^\top A V.$ The
$(k+1)\times (k+1)$ matrix $T$ denotes the operator $A$ restricted to
$\calK,$ expressed in the basis $\{v_i\}_{i=0}^k.$ Now, since $v, \;
Av \in \calK,$ we have $Av = (VV^\top)A(VV^\top)v = VTV^\top v.$
Iterating this argument, we obtain that for all $i \le k$, we have
$A^iv = VT^i V^\top v$ and, hence, by linearity, $p(A)v = Vp(T)V^\top
v$ for any $p \in \Sigma_k.$ Also, note that for every $w \in \calK,$
we have $w = VV^\top w,$ and hence, \[w^\top Aw = (w^\top VV^\top) A
(V V^\top w) = w^\top V (V^\top A V) V^\top w = (w^\top V) T (V^\top
w).\] In words, the above equality says that for any vector $w \in
\calK,$ the Rayleigh quotient of the vector $w$ with respect to $A$ is
the same as the Rayleigh quotient of the vector $V^\top w$ with
respect to $T.$

\vspace{2mm}
The Lanczos method computes the largest eigenvalue of $T,$ $\lambda_1(T)$
and outputs it as an approximation to $\lambda_1(A).$ By the variational characterization of the largest eigenvalues, it follows that $\lambda_1(T) \le \lambda_1(A).$ We have 
\[\lambda_1(T) =
\max_{w \in \rea^{k+1}} \frac{w^\top T w}{w^\top w}= \max_{z \in
  \calK} \frac{z^\top VTV^\top z}{z^\top z}
= \max_{z \in
  \calK} \frac{z^\top A z}{z^\top z}
= \max_{p \in \Sigma_{k}} \frac{v^\top p(A)
  A p(A) v}{v^\top p(A)^2 v} = \max_{p \in \Sigma_{k}} \frac{\sum_i
  \lambda_i p(\lambda_i)^2 \alpha_i^2 }{\sum_i p(\lambda_i)^2
  \alpha_i^2},\] 
where the second equality holds since 
$\calK$ is a $k+1$ dimensional subspace with the columns of $V$ as an
orthonomal basis, and the fourth equality holds because every $z \in
\calK$ can be expressed as $p(A)V$ for some $p \in \Sigma_{k}.$

\vspace{2mm}
Since $v$ is picked uniformly at random, with probability at least
$\nfrac{1}{2},$ we have $\alpha_1^2 \ge \nfrac{1}{4n}.$ Thus, assuming
that $\alpha_1^2 \ge \nfrac{1}{4n},$ for any $p \in \Sigma_{k},$ we
can bound the relative error:
\[\frac{\lambda_{1}(A)-\lambda_{1}(T)}{\lambda_{1}(A)} \le  \frac{\sum_{i=0}^n (1 -
  \nfrac{\lambda_i}{\lambda_1}) p(\lambda_i)^2 \alpha_i^2 }{\sum_{i=0}^n
  p(\lambda_i)^2 \alpha_i^2} \le \delta +\frac{\sum_{\lambda_i <
    (1-\delta)\lambda_1} p(\lambda_i)^2 \alpha_i^2 }{p(\lambda_1)^2
  \alpha_1^2} \le \delta + 4n \sup_{\lambda \in
  [0,(1-\delta)\lambda_1]} \frac{p(\lambda)^2}{p(\lambda_1)^2},\]
where the second inequality follows by splitting the sum in the
numerator depending on whether $\lambda \ge (1-\delta)\lambda_1,$ or
otherwise.

\vspace{2mm}
Observe that if we pick the polynomial $p(\lambda) =
\left(\nfrac{\lambda}{\lambda_1}\right)^s$ for $s \defeq \ceil{
  \nfrac{1}{2\delta} \cdot \log \nfrac{4n}{\delta}}$ in the above
bounds,  the relative error is bounded by $O(\delta).$
Hence, the Lanczos method after $k = O(\nfrac{1}{\delta} \cdot \log
\nfrac{n}{\delta})$ iterations finds a vector with Rayleigh
quotient at least $(1-O(\delta))\lambda_1$ with constant probability, essentially matching the guarantee of the 
power method.

\vspace{2mm}
However, we use the polynomial approximations $p_{s,d}$ to $x^s$ from
Section~\ref{sec:xk-approx} to show that the Lanczos method can do
better. We use the polynomial $p(\lambda) =
p_{s,d}\left(\nfrac{\lambda}{\lambda_1}\right)$ for $s = \ceil{
  \nfrac{1}{2\delta} \cdot \log \nfrac{4n}{\delta}}$ as above, and $d
= \ceil{ \sqrt{2s \cdot \log \nfrac{2n}{\delta}}}.$ In this case, we
know that for all $\lambda$ such that $|\lambda| \le \lambda_1,$ we
have $|p(\lambda) - \left(\nfrac{\lambda}{\lambda_1}\right)^s| \le
\nfrac{\delta}{n}.$ Hence, $p(\lambda_1) \ge 1 - \nfrac{\delta}{n},$
and 
$$\sup_{\lambda \in [0,(1-\delta)\lambda_1]} p(\lambda)^2 \le
\sup_{\lambda \in [0,(1-\delta)\lambda_1]}
\left(\nfrac{\lambda}{\lambda_1}\right)^{2s} + \nfrac{\delta}{n} =
O\left(\nfrac{\delta}{n}\right).$$ Since the degree of this polynomial
is $d,$ we obtain that $d = O\left( \nfrac{1}{\sqrt{\delta}} \cdot
  \log \nfrac{n}{\delta} \right)$ iterations of Lanczos method suffice
to find a vector with Rayleigh quotient at least
$(1-O(\delta))\lambda_1.$

\vspace{2mm}
It remains to analyze the time taken by this algorithm to compute $\lambda_1(T).$ Let $t_A$ denote
the number of operations required to compute $Au,$ given a vector $u.$
We first describe how to quickly compute an orthonormal basis for $\calK$. The procedure is essentially the same as  the one used in the Conjugate
Gradient method. We iteratively compute $Av_{i},$ orthogonalize it with respect to 
$v_{i},\ldots,v_{0},$ and scale it to norm 1 in order to obtain $v_{i+1}.$ As
in the Conjugate Gradient method, we have $Av_j \in
\Span\{v_0,\ldots,v_{j+1}\}$ for all $j < k$ and, hence, using the
symmetry of $A,$ we obtain, $v_j^\top (Av_i) = v_i^\top (Av_j) = 0$ for $j+1 < i.$
Thus, we need to orthogonalize $Av_i$ only with respect to  $v_{i}$ and $v_{i-1}.$
This also implies that $T$ is tridiagonal. Hence, we can construct $V$
and $T$ using $O((t_A+n)k)$ operations. (Note the subtle difference; here, we ensure the basis vectors are orthonormal, instead of
$A$-orthogonal as in the case of Conjugate Gradient.) The only
remaining step is to compute the largest eigenvalue of $T,$ which can
be found via an eigendecomposition of $T.$ Since $T$ is
tridiagonal, this step can be upper bounded by $O(k^2)$
(see~\cite{PanC99}). Thus, we have the following theorem:

\begin{theorem}
  Given a symmetric PSD matrix $A,$ and a parameter $\delta > 0,$ the
  Lanczos method after $k$ iterations, for $k = O\left(
    \nfrac{1}{\sqrt{\delta}} \cdot \log \nfrac{n}{\delta} \right),$
  outputs a value $\mu \in [(1-\delta)\lambda_1(A), \lambda_1(A)]$
  with constant probability over the choice of random $v$. The total
  number of operations required is $O((t_A +n)k + k^2),$ where $t_A$
  is the number of operations required to multiply $A$ with a given
  vector.
\end{theorem}

\noindent
The eigenvector $w$ of  $T$ which achieves $\lambda_1(T)$ can be used to give a candidate for the the largest eigenvector of $A,$ {\em i.e.}, the vector $Vw.$

\paragraph{Beyond the largest eigenvalue.}
The Lanczos method can also be used to approximate several large
eigenvalues of $A.$ The algorithm is essentially the same, except that
we choose a Krylov subspace of higher order $k,$ and output the top
$r$ eigenvalues of the matrix $T.$ Using techniques similar to above,
we can achieve a similar speed-up in the case where the top
eigenvalues of $A$ are well-separated. Such results were obtained
in~\cite{Kaniel66, Saad80} (see~\cite[Chapter 6]{Saad11} and the notes
therein).
\begin{theorem}
  Given a symmetric PSD matrix $A$ with eigenvalues $\lambda_1 \ge
  \cdots \ge \lambda_n$ such that $|\lambda_i - \lambda_{i+1}| \ge
  \delta \lambda_1$ for $i=1,\ldots,r,$ and a parameter $\delta,$
  after $k = O\left( \nfrac{r}{\sqrt{\delta}} \cdot \log
    \nfrac{nr}{\delta} \right)$ iterations of the Lanczos method, the
  matrix $T$ will have $r$ largest eigenvalues $\mu_1 \ge \cdots \ge
  \mu_r$ such that $\mu_i \in
  [(1-\nfrac{\delta}{3})\lambda_i,\lambda_i],$ with constant
  probability over the choice of random $v.$ The total number of
  operations required is $O((t_A +n)k + k^2),$ where $t_A$ is the
  number of operations required to multiply $A$ with a given
  vector.
\end{theorem}

\noindent
Letting $u_1,\ldots,u_r$ denote the eigenvectors of $T$ corresponding
to eigenvalues $\mu_1,\ldots,\mu_r,$ we obtain $Vu_1, \ldots, Vu_r$ as
the candidate eigenvectors, as before.

\paragraph{Computing $f(A)v.$}
The Lanczos method can in fact be used more generally to obtain a fast
approximation to $f(A)v$ for any function $f,$ and any vector $v.$ We
saw that if we work with the Krylov subspace $\{v,Av, \ldots,A^kv\},$
for any polynomial $p \in \Sigma_k,$ we have $Vp(T)V^\top v = p(A)v.$
Hence, a natural approximation for $f(A)v$ is $Vf(T)V^\top v.$
Moreover, using the method above, the number of operations required is
$O((t_A + n)k)$ plus those required to compute $f(\cdot)$ on the
$(k+1)\times (k+1)$ tridiagonal matrix $T$, which can usually be upper
bounded by $O(k^2)$ via diagonalization (see~\cite{PanC99}). Letting
$\calI \defeq [\lambda_{n}(A),\lambda_{1}(A)],$ the error in the
approximation can be upper bounded by $2\eps_{f,\calI}(k)$, the
uniform approximation error achieved by the best degree $k$ polynomial
approximating $f$ on $\calI$ (see~\cite[Chapter 19]{Vishnoi12}). This
method derives its power from the fact that, in order to compute a good
approximation, just the {\em existence} of a good polynomial that approximates
$f$ on $\calI$ is sufficient, and we do not need to \emph{know} the polynomial.

\subsection{Computing the Matrix Exponential}
\label{sec:applications:matrix-exp}

In this section we consider the problem of computing $\exp(-A)v$ for
an $n \times n$ PSD matrix $A$ and a vector $v.$ Recall that $\exp(-A)
= \sum_{k=0}^\infty \frac{(-1)^k A^k}{k!}.$ Of particular interest is
the special case $\exp(-s(I-W)) = e^{-s} \sum_{k \ge 0} \frac{s^k}{k!}
W^k$ where $W$ is the random walk matrix associated to a graph
$G=(V,E)$ defined in Section \ref{sec:applications:sim-rw}. In terms
of the normalized Laplacian $\mathcal{L}=I-W,$ this is the same as
$\exp(-s \mathcal{L}).$ This matrix corresponds to the transition
matrix of a {\em continuous-time} random walk of length $s$ on $G,$
also called the {\em heat-kernel} walk on $G,$ see
\cite{Chung97,LevinPW06}.  Note that this walk can be interpreted as
the distribution of a discrete-time random walk after a
Poisson-distributed number of steps with mean $s$ since $\exp(-s
\mathcal{L})= e^{-s} \sum_{k \ge 0} \frac{s^k W^k}{k!}.$ These random
walks are of importance in probability and algorithms, and the ability
to simulate them in time near-linear in the number of edges in the
graph results in near-linear time algorithms for problems such as the
{\em balanced} version of the Sparsest Cut problem introduced in
Section \ref{sec:sparsestcut}. More generally, fast computation of
$\exp(-A)v$ plays a crucial role, via the Matrix Multiplicative Weights
Update method, in obtaining fast combinatorial algorithms to solve
semi-definite programs, see \cite{AroraK07, Orecchia11,AHK12}.

\vspace{2mm}
The most natural way to compute $\exp(-A)v$ is  to approximate the matrix exponential using the Taylor series approximation for the exponential, or to use the improved
polynomial approximations constructed in
Section~\ref{sec:exp-poly-ub}. Indeed, Theorem \ref{thm:exp-poly-ub} can be used to compute  a $\delta$ approximation to $\exp(-A)v$  in time $O\left(\left(t_A+n\right) \sqrt{\|A\| \log
\nfrac{1}{\delta}}\right);$ similar to  Theorem \ref{thm:simrw}. 
However, Theorem~\ref{lem:exp-poly-lb} implies that 
that no polynomial approximation can get rid of the dependence on $\sqrt{\|A\|}$ in the running time above.

\vspace{2mm} What about rational approximations to $e^{-x}$ proved in
Section \ref{sec:rational-apx}? Indeed, we can use the rational
approximation from Theorem \ref{thm:exp-rational} to obtain $\|
\exp(-A) - \left( S_d(A)\right)^{-1}\| \leq 2^{-\Omega(d)}$, where $ \left( S_d(A)\right)^{-1}v$ is the approximation to $\exp(-A)v.$ For most applications an error of $\delta = \nfrac{1}{{\rm poly}(n)}$ suffices, so it is sufficient to choose $d=O(\log n).$   How do we compute $ \left(
  S_d(A)\right)^{-1}v$? Clearly, inverting $S_d(A)$ is not a good idea
since that would be at least as inefficient as matrix inversion.  The
next natural idea is to factor $S_d(x) = \alpha_0 \prod_{i=1}^d
(x-\beta_i)$ and then calculate $(S_d(A))^{-1}v = \alpha_0
\prod_{i=1}^d (A-\beta_i I)^{-1} v.$ Since $d$ is  small, namely $O(\log n),$  the cost of computing $(S_d(A))^{-1}v$ reduces to the cost of computing $(A-\beta_i
I)^{-1}u_i$. Thus, it is suffices to speed a computation of this form. 
The first problem is that $\beta_i$s could be
complex, as is indeed the case for the polynomial
$S_d$ as discussed in Section \ref{sec:rational-apx}. However, since  $S_d$ has real coefficients,  its complex roots appear as
conjugates. Hence, we can combine the factors corresponding to the pairs and reduce the task to computing $(A^2 -
(\beta_i + \bar{\beta_i})A + |\beta_i|^2 I )^{-1}u.$ 
The matrix $(A^2 -
(\beta_i + \bar{\beta_i})A + |\beta_i|^2 I )$ is easily seen to be PSD and we can try to apply the Conjugate gradient method to compute $(A^2 -
(\beta_i + \bar{\beta_i})A + |\beta_i|^2 I )u.$ 
However, the condition number of this matrix can be comparable to that of $A$, which gives no significant advantage over $\sqrt{\|A\|}$. To see this, observe that
$|\beta_i| \le d$ (see~\cite{Zemyan05}), and consider a matrix $A$
with $\lambda_1(A) \gg d,$ and $\lambda_n(A)=1.$ For such a matrix,
the condition number of $(A^2 - (\beta_i + \bar{\beta_i})A +
|\beta_i|^2 I )$ is $\Omega(\nfrac{\lambda^2_1(A)}{d^2}),$ which is approximately the square of  the condition number of $A$ for small $d.$ 

\vspace{2mm}
Similarly, the rational approximations to $e^{-x}$ in
Section~\ref{sec:exp-rational2} suggest the vector
$p_d((I+\nfrac{A}{d})^{-1})v$ as an approximation to $\exp(-A)v,$
where $p_d$ is the polynomial given by Theorem~\ref{thm:SSV}. Once again,
for any PSD matrix $A,$ though the matrix $(I+\nfrac{A}{d})$ is PSD, the condition number of $(I+\nfrac{A}{d})$ could be comparable to that of $A$. 
 Hence for arbitrary PSD matrices, the
rational approximations to $e^{-x}$ seem insufficient for obtaining
improved algorithms for approximating the matrix exponential.
Indeed, $O\left(\left(t_A+n\right) \sqrt{\|A\| \log
\nfrac{1}{\delta}}\right)$   is the best  result  known for computing the matrix exponential-vector product for a general PSD matrix $A$, see \cite{OrecchiaSV12}.

\vspace{2mm}
The above approach of using rational approximations shows how to reduce the computation of $\exp(-A)v$ to a small number of  linear systems involving the matrix $A.$  For an important special class of matrices, we can exploit the fact that
there exist algorithms that are much faster than Conjugate Gradient and allow us to approximate
$(I+\nfrac{A}{d})^{-1}u,$ for a given $u.$ In particular, for
a {\em symmetric and diagonally dominant} (SDD) matrix\footnote{A matrix $A$ is said to be Symmetric and
  Diagonally Dominant (SDD) if it is symmetric, and for all $i,$
  $A_{ii} \ge \sum_{j \neq i} |A_{ij}|.$} $A$, there are powerful
near-linear-time SDD system solvers~\cite{SpielmanT04, KoutisMP11,
  KelnerOSZ13} whose guarantees are given in the following theorem. 

\begin{theorem}
\label{thm:sdd-solver}
  Given an $n \times n$ SDD matrix $A$ with $m$ non-zero entries, a
  vector $v$, and $\delta_1 > 0$, there is an algorithm that, in
  $\tilde{O}\left(m \log \nfrac{1}{\delta_1}\right)$ time, 
  computes a vector $u$ such that $\|u-A^{-1}v\|_A \le \delta_1
  \|A^{-1}v\|_A\ .$
Moreover, $u=Zv$ where $Z$ depends on $A$ and
  $\delta_1,$ and is such that $(1-\delta_1)A^{-1} \preceq Z \preceq
  (1+\delta_1)A^{-1}.$
\end{theorem}

\noindent
At the end of this section, we show how to compute the
coefficients of $p_d$ from Theorem \ref{thm:SSV}  efficiently, and show that each coefficient is
bounded by $d^{O(d)}.$ Assuming this we show that we can compute $p_d((I+\nfrac{A}{d})^{-1})v$ as an approximation to $\exp(-A)v$ in near-linear time using Theorem \ref{thm:sdd-solver}. 
Note that if $A$ is SDD, so is $(I+\nfrac{A}{d}).$  
We repeatedly use the SDD
solver of Theorem \ref{thm:sdd-solver} to approximate $(I+\nfrac{A}{d})^{-i}v,$ for all
$i=1,\ldots,d,$ and let $Z$ denote the linear operator such that the
SDD solver returns the vector $Z^iu$ as the approximation. Let $B
\defeq (I+\nfrac{A}{d})^{-1}.$ From the guarantee on the SDD solver from the theorem above,
we know that $-\delta_1 B \preceq Z - B \preceq \delta_1 B.$ Applying the 
triangle inequality to the identity $Z^i - B^i = \sum_{j=0}^{i-1}
Z^{i-1-j} (Z-B)B^{j},$ and using $\norm{B} \le 1,$ we obtain, $\norm{Z^i
  - B^i} \le \delta_1\cdot i(1+\delta_1)^i.$ Thus,
$\norm{p_d(Z)-p_d(B)} \le d^{O(d)} \cdot \delta_1(1+\delta_1)^d.$
Hence, we can choose $\delta_1 = \delta\cdot d^{-\Theta(d)}$ for the
SDD solver in order for the final approximation to have error at most
$\delta.$ Since $d = \Theta(\log \nfrac{1}{\delta})$ suffices, this
results in an overall running time of $\tilde{O}(m).$
We summarize the result in the following theorem.
\begin{theorem}
\label{thm:exp}
There is an algorithm that, given an SDD matrix $A$ with $m$ non-zero
entries, a vector $v,$ and $\delta \in (0,1] $, computes a vector $u$
such that $\norm{\exp(-A)v-u} \le \delta \|v\|$ in time
$\tilde{O}((m+n)\log(2+\norm{A}) {\rm.
polylog}\;\nfrac{1}{\delta})$.
\end{theorem}
\noindent
The above theorem was first proved in~\cite{OrecchiaSV12}. However, instead
of  computing the coefficients of the polynomial $p_d$
explicitly, the authors in~\cite{OrecchiaSV12} appealed to  the Lanczos method
from numerical linear algebra that allows them to achieve the error
guarantee of the approximating polynomial {\em without} explicit knowledge
of the polynomial. 

\vspace{2mm}
Coming back to graphs,  an important corollary of this theorem is that $\exp(-s\mathcal{L})v,$ the distribution after an $s$-length continuous time random walk on the graph with normalized Laplacian $\mathcal{L}$ starting with a distribution $v,$ can be approximately computed in $\tilde{O}(m \log s)$ time. Recall that for simple random walks, from Section \ref{sec:rw}, we do not know how to do better than $O(m \sqrt{s})$ time.

\paragraph{Computing the coefficients of $p_d$.}
We now address the issue of explicitly computing the coefficients of
$p_d.$  It suffices to compute them to a precision of $2^{-{\rm poly}(d)}$ and we present the salient steps. Recall that  in the proof of
Theorem~\ref{thm:exp-rational}, the polynomial
$r_{d-1}(t)$ that minimizes $\int_{-1}^1 \left(
  f^{(1)}_d(t)-r_{d-1}(t)\right)^2 \dif{t}$ (see
Equation \eqref{eq:exp-rational:l2-min}) is given by $r_{d-1}(t) =
\sum_{k=0}^{d-1} \sqrt{2k+1}\cdot \gamma_k \cdot L_k(t).$ The
Legendre polynomials can be written as $L_k(t) = 2^{-k}
\sum_{i=0}^{\floor{\nfrac{k}{2}}} x^{k-2i}\cdot \binom{k}{i}
\binom{2k-2i}{k}$, see \cite[Chapter 22]{AbramowitzS64}. Thus, assuming we
know $\{\gamma_k\}_{k=0}^{d-1},$ we can compute the coefficients of
$r_{d-1}$ in $\poly(d)$ operations, and the sizes of the coefficients
of $r_{d-1}$ can only be $2^{O(d)}$ larger. Since $q_d(y) = \int_{y}^1
r_{d-1}(t) \dif{t},$ given the coefficients of $r_{d-1},$ we can simply integrate in order to find the coefficients of $q_d$. The
approximating polynomial $p_d(x)$ is given by $p_d(x) \defeq
q_d(1-2x).$ Hence, given the coefficients of $q_d,$ those of $p_d$ can
be calculated in poly$(d)$ operations, and again can only be at
most $2^{O(d)}$ larger. Hence, it suffices to show how to compute
$\{\gamma_k\}_{k=0}^{d-1}.$ 

With the substitution $z=d(1+v)$ in
Equation~\eqref{eq:exp-rational:gamma}, 
we have
$\gamma_k = - d
\int_{0}^\infty \left( \tfrac{z}{z+1}\right)^k e^{-dz}
G_k(d(1+z))\dif{z}$. 
The Laguerre polynomials (of order $1$)
$G_k$ are explicitly given as $G_k(t) = \sum_{i=0}^k (-1)^i
\binom{k+1}{k-i} \frac{t^i}{i!}$~\cite[Chapter 22]{AbramowitzS64}.
After using this expansion for $G_k,$ it suffices to compute
the integrals $\int_{0}^\infty \tfrac{z^i}{(z+1)^j} e^{-dz} \dif{z}$
for $0 \le j \le i \le d.$ If we know the values of these integrals,
we can compute $\gamma_k$s in $\poly(d)$ operations, though
the coefficients may now increase by a factor of $d^{O(d)}.$ For any
$0 \le j \le i \le d,$ substituting $w = z+1,$ we obtain
$\int_{0}^\infty \tfrac{z^i}{(z+1)^j} e^{-dz} \dif{z} = e^{-d}
\int_{1}^\infty \tfrac{(w-1)^i}{w^j} e^{-dw} \dif{w}.$ Since we can
expand $(w-1)^i$ using the Binomial theorem, it suffices to compute
integrals of the form
 $\int_{1}^\infty w^{-j} e^{-dw} \dif{w}$ for 
$-d
\le j \le d,$ where again we lost at most $2^d$ in the magnitude of
the numbers. 
For $j \leq 0,$ this is a simple integration. For $j \geq 1,$ the integral can be expressed using  the Exponential Integral \cite{ExponentialIntegral}. Hence, it has the following rapidly convergent power series for $d >1$, which can be used both to compute $E_j(d)$s
and bound them easily:
$$  E_j(d) =\int_{1}^\infty w^{-j} e^{-dw} \dif{w}=  \frac{e^{-d}}{d} \sum_{k=0}^\infty \frac{(-1)^k (j+k-1)!}{(j-1)!d^k},$$
see \cite{ExponentialIntegral}.
Combining everything, the coefficients of $p_d$ 
can be approximated up to $d^{-\Theta(d)}$ error in time $\poly(d)$ using
$\poly(d)$ sized registers.

\subsection{Matrix Inversion via Exponentiation}
\label{sec:applications:matrix-inverse-reduction}
Our final application of approximation theory is a rather surprising result which reduces a computation of the form $A^{-1}v$ for a PSD $A$, to the computation of a small number of terms of the form $\exp(-sA)v.$ One way to interpret this result is that the linear system solvers deployed in the previous section are {\em necessary}. The other way is to see this as a new approach to speed up computations beyond the Conjugate Gradient method to compute $A^{-1}v$ for PSD matrices, a major open problem in numerical linear algebra with implications far beyond. 

\vspace{2mm}
This result is an immediate corollary of 
Theorem~\ref{thm:inverse-approx}, proved in Section~\ref{sec:exp-reduction}, which shows than we can approximate $x^{-1}$ with a sum of a small number of exponentials, where the approximation is valid for all $x \in [\delta,1].$ 
\begin{theorem}[Corollary to Theorem~\ref{thm:inverse-approx}, \cite{SachdevaV13}]
\label{thm:matrix-inverse-approx}
Given $\eps,\delta \in (0,1],$  there exist $\poly( \log
\nfrac{1}{\eps\delta})$ numbers $0< w_j,t_j =
O(\poly(\nfrac{1}{\eps\delta})),$ such that for all symmetric matrices
$A$ satisfying $\eps I \preceq A \preceq I,$ we have $(1-\delta)A^{-1}
\preceq \sum_j w_j e^{-t_j A} \preceq (1+ \delta)A^{-1}.$
\end{theorem}

\noindent
Since the above reduction only requires that the matrix $A$ be
positive-definite, it immediately suggests an approach to
approximating $A^{-1}v$: Approximate $e^{-t_j A}v$ for each
$j$ and return the vector $\sum_j w_j e^{-t_j A}v$ as an approximation
for $A^{-1}v.$ Since the weights $w_j$ are
$O(\poly(\nfrac{1}{\delta\eps})),$ we lose only a polynomial factor in
the approximation error.

\paragraph{Acknowledgments.} We would like to thank Elisa Celis and
Oded Regev for useful discussions. We would also like to thank
L\'{a}szl\'{o} Babai for helpful comments.

{
\bibliographystyle{plain}
\bibliography{papers}}

\begin{thebibliography}{10}

\bibitem{Aaronson08}
S.~Aaronson.
\newblock The polynomial method in quantum and classical computing.
\newblock In {\em Foundations of Computer Science, 2008. FOCS '08. IEEE 49th
  Annual IEEE Symposium on}, pages 3--3, 2008.

\bibitem{AbramowitzS64}
M.~Abramowitz and I.A. Stegun.
\newblock {\em Handbook of Mathematical Functions}.
\newblock Dover, New York, fifth edition, 1964.

\bibitem{AlonM85}
Noga Alon and V.~D. Milman.
\newblock $\lambda_{1}$, isoperimetric inequalities for graphs, and
  superconcentrators.
\newblock {\em J. Comb. Theory, Ser. B}, 38(1):73--88, 1985.

\bibitem{Andersson81}
Jan-Erik Andersson.
\newblock Approximation of $e^{-x}$ by rational functions with concentrated
  negative poles.
\newblock {\em Journal of Approximation Theory}, 32(2):85 -- 95, 1981.

\bibitem{AHK12}
Sanjeev Arora, Elad Hazan, and Satyen Kale.
\newblock The multiplicative weights update method: a meta-algorithm and
  applications.
\newblock {\em Theory of Computing}, 8(6):121--164, 2012.

\bibitem{AroraK07}
Sanjeev Arora and Satyen Kale.
\newblock A combinatorial, primal-dual approach to semidefinite programs.
\newblock In {\em STOC}, pages 227--236, 2007.

\bibitem{Beals01}
Robert Beals, Harry Buhrman, Richard Cleve, Michele Mosca, and Ronald de~Wolf.
\newblock Quantum lower bounds by polynomials.
\newblock {\em J. ACM}, 48(4):778--797, July 2001.

\bibitem{BeigelRS95}
R.~Beigel, N.~Reingold, and D.~Spielman.
\newblock {PP} is closed under intersection.
\newblock {\em Journal of Computer and System Sciences}, 50(2):191 -- 202,
  1995.

\bibitem{BM05}
Gregory Beylkin and Lucas Monz\'{o}n.
\newblock On approximation of functions by exponential sums.
\newblock {\em Applied and Computational Harmonic Analysis}, 19(1):17 -- 48,
  2005.

\bibitem{BM10}
Gregory Beylkin and Lucas Monz\'{o}n.
\newblock Approximation by exponential sums revisited.
\newblock {\em Applied and Computational Harmonic Analysis}, 28(2):131 -- 149,
  2010.
\newblock Special Issue on Continuous Wavelet Transform in Memory of Jean
  Morlet, Part I.

\bibitem{Borel1905}
{\'E}mile Borel.
\newblock {\em Lecons sur les Fonctions de Variables R{\'e}elles et les
  D{\'e}veloppements en S{\'e}ries de Polynomes}.
\newblock Gauthier-Villars, Paris (2nd edition, 1928), 1905.

\bibitem{BunT13}
Mark Bun and Justin Thaler.
\newblock Dual lower bounds for approximate degree and {Markov-Bernstein}
  inequalities.
\newblock In {\em Automata, Languages, and Programming}, volume 7965 of {\em
  Lecture Notes in Computer Science}, pages 303--314. Springer Berlin
  Heidelberg, 2013.

\bibitem{Chebyshev1854}
P.~L. Chebyshev.
\newblock Th\'eorie des m\'ecanismes connus sous le nom de parall\'elogrammes.
\newblock {\em M\'em. Acad. Sci. P\'etersb.}, 7:539--568, 1854.

\bibitem{Chebyshev1859}
P.~L. Chebyshev.
\newblock Sur les questions de minima qui se rattachent \`a la repr\'esentation
  approximative des fonctions.
\newblock {\em M\'em. Acad. Sci. P\'etersb.}, 7:199--291, 1859.

\bibitem{Cheeger70}
J.~Cheeger.
\newblock {A lower bound for the smallest eigenvalue of the Laplacian}.
\newblock {\em Problems Anal.}, pages 195--199, 1970.

\bibitem{Cheney66}
E.~W. Cheney.
\newblock {\em Introduction to approximation theory}.
\newblock McGraw-Hill, New York :, 1966.

\bibitem{Chung97}
Fan~R.K. Chung.
\newblock {\em Spectral Graph Theory (CBMS Regional Conference Series in
  Mathematics, No. 92)}.
\newblock American Mathematical Society, 1997.

\bibitem{CodyMV69}
W.J Cody, G~Meinardus, and R.S Varga.
\newblock Chebyshev rational approximations to $e^{-x}$ in $[0, \infty)$ and
  applications to heat-conduction problems.
\newblock {\em Journal of Approximation Theory}, 2(1):50 -- 65, 1969.

\bibitem{GoncharR83}
A.A. Gonchar and E.A. Rakhmanov.
\newblock On convergence of simultaneous {Pad\'e} approximants for systems of
  functions of {Markov} type.
\newblock {\em Proc. Steklov Inst. Math.}, 157:31--50, 1983.

\bibitem{GrahamKP94}
Ronald~L. Graham, Donald~E. Knuth, and Oren Patashnik.
\newblock {\em Concrete Mathematics: A Foundation for Computer Science}.
\newblock Addison-Wesley Longman Publishing Co., Inc., Boston, MA, USA, 2nd
  edition, 1994.

\bibitem{Grover96}
Lov~K. Grover.
\newblock A fast quantum mechanical algorithm for database search.
\newblock In {\em Proceedings of the twenty-eighth annual ACM symposium on
  Theory of computing}, STOC '96, pages 212--219, New York, NY, USA, 1996. ACM.

\bibitem{HestenesS52}
Magnus~R. Hestenes and Eduard Stiefel.
\newblock Methods of {C}onjugate {G}radients for solving linear systems.
\newblock {\em Journal of Research of the National Bureau of Standards},
  49:409--436, December 1952.

\bibitem{HochbruckL97}
Marlis Hochbruck and Christian Lubich.
\newblock On {K}rylov subspace approximations to the matrix exponential
  operator.
\newblock {\em SIAM J. Numer. Anal.}, 34(5):1911--1925, October 1997.

\bibitem{JainJUW11}
Rahul Jain, Zhengfeng Ji, Sarvagya Upadhyay, and John Watrous.
\newblock {QIP}$=${PSPACE}.
\newblock {\em J. ACM}, 58(6):30:1--30:27, December 2011.

\bibitem{JainUW09}
Rahul Jain, Sarvagya Upadhyay, and John Watrous.
\newblock Two-message quantum interactive proofs are in {PSPACE}.
\newblock In {\em Proceedings of the 2009 50th Annual IEEE Symposium on
  Foundations of Computer Science}, FOCS '09, pages 534--543, Washington, DC,
  USA, 2009. IEEE Computer Society.

\bibitem{JainW09}
Rahul Jain and John Watrous.
\newblock Parallel approximation of non-interactive zero-sum quantum games.
\newblock {\em 2012 IEEE 27th Conference on Computational Complexity},
  0:243--253, 2009.

\bibitem{Kthesis}
Satyen Kale.
\newblock Efficient algorithms using the multiplicative weights update method.
\newblock Technical report, Princeton University, Department of Computer
  Science, 2007.

\bibitem{ExponentialIntegral}
N.N. Kalitkin and I.A. Panin.
\newblock On the computation of the exponential integral.
\newblock {\em Mathematical Models and Computer Simulations}, 1(1):88--90,
  2009.

\bibitem{Kaniel66}
Shmuel Kaniel.
\newblock Estimates for some computational techniques in linear algebra.
\newblock {\em Math. Comp.}, 20:369--378, 1966.

\bibitem{KelnerOSZ13}
Jonathan~A. Kelner, Lorenzo Orecchia, Aaron Sidford, and Zeyuan~Allen Zhu.
\newblock A simple, combinatorial algorithm for solving {SDD }systems in
  nearly-linear time.
\newblock In {\em Proceedings of the 45th annual ACM symposium on Symposium on
  theory of computing}, STOC '13, pages 911--920, New York, NY, USA, 2013. ACM.

\bibitem{KoutisMP11}
Ioannis Koutis, Gary~L. Miller, and Richard Peng.
\newblock A nearly- $m\log n$ time solver for {SDD} linear systems.
\newblock In {\em Proceedings of the 2011 IEEE 52nd Annual Symposium on
  Foundations of Computer Science}, FOCS '11, pages 590--598, Washington, DC,
  USA, 2011. IEEE Computer Society.

\bibitem{KuczynskiW92}
J.~Kuczy\'{n}ski and H.~Wo\'{z}niakowski.
\newblock Estimating the largest eigenvalues by the power and {L}anczos
  algorithms with a random start.
\newblock {\em SIAM J. Matrix Anal. Appl.}, 13(4):1094--1122, October 1992.

\bibitem{Lanczos52}
Cornelius Lanczos.
\newblock Solution of systems of linear equations by minimized iterations.
\newblock {\em J. Res. Natl. Bur. Stand}, 49:33--53, 1952.

\bibitem{LevinPW06}
David~A. Levin, Yuval Peres, and Elizabeth~L. Wilmer.
\newblock {\em {Markov chains and mixing times}}.
\newblock American Mathematical Society, 2006.

\bibitem{Markov1890}
Andrei~Andreyevich Markov.
\newblock Ob odnom voproce {D.I.} {M}endeleeva.
\newblock {\em Zapiski Imperatorskoi Akademii Nauk SP6}, 62:1--24, 1890.

\bibitem{MarkovVA1892}
V.~A. Markov.
\newblock O funktsiyakh, naimeneye uklonyayushchikhsya ot nulya vdannom
  promezhutke.
\newblock 1892.

\bibitem{Mihail89}
Milena Mihail.
\newblock Conductance and convergence of {M}arkov chains-{A} combinatorial
  treatment of expanders.
\newblock In {\em FOCS}, pages 526--531, 1989.

\bibitem{Mitzenmacher2005}
Michael Mitzenmacher and Eli Upfal.
\newblock {\em Probability and computing: Randomized algorithms and
  probabilistic analysis}.
\newblock Cambridge University Press, New York, NY, USA, 2005.

\bibitem{Newman64}
Donald~J. Newman.
\newblock Rational approximation to {$\vert x\vert $}.
\newblock {\em Michigan Math. J.}, 11:11--14, 1964.

\bibitem{Newman74}
Donald~J. Newman.
\newblock Rational approximation to $e^{-x}$.
\newblock {\em Journal of Approximation Theory}, 10(4):301 -- 303, 1974.

\bibitem{NisanS94}
Noam Nisan and Mario Szegedy.
\newblock On the degree of {B}oolean functions as real polynomials.
\newblock {\em computational complexity}, 4(4):301--313, 1994.

\bibitem{Orecchia11}
Lorenzo Orecchia.
\newblock {\em Fast Approximation Algorithms for Graph Partitioning using
  Spectral and Semidefinite-Programming Techniques}.
\newblock PhD thesis, EECS Department, University of California, Berkeley, May
  2011.

\bibitem{OrecchiaSV11-arxiv}
Lorenzo Orecchia, Sushant Sachdeva, and Nisheeth~K. Vishnoi.
\newblock Approximating the exponential, the {L}anczos method and an
  $\tilde{O}$(m)-time spectral algorithm for {Balanced} {Separator}.
\newblock {\em CoRR}, abs/1111.1491, 2011.

\bibitem{OrecchiaSV12}
Lorenzo Orecchia, Sushant Sachdeva, and Nisheeth~K. Vishnoi.
\newblock Approximating the exponential, the {Lanczos} method and an
  $\tilde{O}$(m)-time spectral algorithm for {B}alanced {Separator}.
\newblock STOC '12, pages 1141--1160, 2012.

\bibitem{OSVV}
Lorenzo Orecchia, Leonard~J. Schulman, Umesh~V. Vazirani, and Nisheeth~K.
  Vishnoi.
\newblock On partitioning graphs via single commodity flows.
\newblock In {\em STOC '08: Proc. 40th Ann. ACM Symp. Theory of Computing},
  pages 461--470, 2008.

\bibitem{OrecchiaV11}
Lorenzo Orecchia and Nisheeth~K. Vishnoi.
\newblock Towards an {SDP}-based approach to spectral methods: A
  nearly-linear-time algorithm for graph partitioning and decomposition.
\newblock In {\em SODA'11: Proc. 22nd Ann. ACM-SIAM Symp. Discrete Algorithms},
  pages 532--545, 2011.

\bibitem{PanC99}
Victor~Y. Pan and Zhao~Q. Chen.
\newblock The complexity of the matrix eigenproblem.
\newblock In {\em STOC'99}, pages 507--516, 1999.

\bibitem{parlett1980symmetric}
Beresford~N Parlett.
\newblock {\em The symmetric eigenvalue problem}, volume~7.
\newblock SIAM, 1980.

\bibitem{Rivlin69}
T.J. Rivlin.
\newblock {\em An introduction to the approximation of functions}.
\newblock Blaisdell book in numerical analysis and computer science. Blaisdell
  Pub. Co., 1969.

\bibitem{rudin-principles}
Walter Rudin.
\newblock {\em Principles of mathematical analysis}.
\newblock McGraw-Hill Book Co., New York, third edition, 1976.
\newblock International Series in Pure and Applied Mathematics.

\bibitem{Saad80}
Y.~Saad.
\newblock On the rates of convergence of the {L}anczos and the block-{L}anczos
  methods.
\newblock {\em SIAM Journal on Numerical Analysis}, 17(5):pp. 687--706, 1980.

\bibitem{Saad11}
Yousef Saad.
\newblock {\em Numerical methods for large eigenvalue problems}.
\newblock Society for Industrial and Applied Mathematics, 2011.

\bibitem{SaadSurvey}
Yousef Saad and Henk~A. van~der Vorst.
\newblock Iterative solution of linear systems in the 20th century.
\newblock {\em Journal of Computational and Applied Mathematics}, 123(1–2):1
  -- 33, 2000.
\newblock Numerical Analysis 2000. Vol. III: Linear Algebra.

\bibitem{SachdevaV13}
S.~{Sachdeva} and N~K. {Vishnoi}.
\newblock Matrix inversion is as easy as exponentiation.
\newblock {\em ArXiv e-prints}, May 2013.

\bibitem{SaffSV75}
E.~B. Saff, A.~{Sch\"{o}nhage}, and R.~S. Varga.
\newblock Geometric convergence to $e^{-z}$ by rational functions with real
  poles.
\newblock {\em Numerische Mathematik}, 25:307--322, 1975.

\bibitem{Schonhage73}
A~{Sch\"{o}nhage}.
\newblock Zur rationalen {A}pproximierbarkeit von $e^{-x}$ {\"{u}ber}
  $[0,\infty)$.
\newblock {\em Journal of Approximation Theory}, 7(4):395 -- 398, 1973.

\bibitem{Sherman09}
Jonah Sherman.
\newblock Breaking the multicommodity flow barrier for ${O}(\sqrt{\log
  n})$-approximations to {S}parsest {C}ut.
\newblock In {\em FOCS'09: Proc. 50th Ann. IEEE Symp. Foundations of Computer
  Science}, 2009.

\bibitem{Sherman13}
Jonah Sherman.
\newblock Nearly maximum flows in nearly linear time.
\newblock In {\em FOCS\rq{}13: Proc. 54th Ann. IEEE Symp. Foundations of
  Computer Science}, 2013.

\bibitem{SherstovThesis}
Alexander Sherstov.
\newblock Lower bounds in communication complexity and learning theory via
  analytic methods.
\newblock Technical report, University of Texas at Austin, 2009.

\bibitem{SpielmanT04}
Daniel~A. Spielman and Shang-Hua Teng.
\newblock Nearly-linear time algorithms for graph partitioning, graph
  sparsification, and solving linear systems.
\newblock In {\em STOC}, pages 81--90, New York, NY, USA, 2004. ACM.

\bibitem{Szego24}
G.~{Szeg\"{o}}.
\newblock \"{U}ber eine {E}igenschaft der {E}xponentialreihe.
\newblock {\em Sitzungsber. Berl. Math. Ges.}, 23:50--64, 1924.
\newblock cited By (since 1996)37.

\bibitem{Szego39}
Gabor Szego.
\newblock {\em Orthogonal polynomials}.
\newblock American Mathematical Society Providence, 4th ed. edition, 1939.

\bibitem{Vishnoi12}
Nisheeth~K Vishnoi.
\newblock ${L}x=b$.
\newblock {\em Foundations and Trends in Theoretical Computer Science},
  8(1-2):1--141, 2012.

\bibitem{Weierstrass1885}
Karl Weierstrass.
\newblock {\"U}ber die analytische {D}arstellbarkeit sogenannter
  willk{\"u}rlicher {F}unctionen einer reellen ver{\"a}nderlichen.
\newblock {\em Sitzungsberichte der K{\"o}niglich Preu{\ss}ischen Akademie der
  Wissenschaften zu Berlin}, 2:633--639, 1885.

\bibitem{Zemyan05}
Stephen~M. Zemyan.
\newblock On the zeroes of the {N}th partial sum of the exponential series.
\newblock {\em The American Mathematical Monthly}, 112(10):pp. 891--909, 2005.

\end{thebibliography}
\end{document}